\documentclass{article}
\usepackage{amsfonts}

\usepackage{amsmath}
\usepackage{pstricks,pst-node,pst-text,pst-3d}
\usepackage{epsf}
\usepackage{bm}

\newtheorem{theorem}{Theorem}

\newtheorem{example}[theorem]{Example}

\newtheorem{remark}[theorem]{Remark}

\newenvironment{proof}[1][Proof]{\noindent\textbf{#1.} }{\ \rule{0.5em}{0.5em}}

\begin{document}

\title{{\large \textbf{Phi-Divergence test statistics for testing the
validity of latent class models for binary data}}}
\author{Felipe$^{1}$, A., Martin$^{2}$, N., Miranda$^{1}$, P. and Pardo$%
^{1}$, L. \\
$^{1}$Department of Statistics and O.R., Complutense University of
Madrid, 28040-Madrid, Spain \\
$^{2}$Department of Statistics, Carlos III University of Madrid, 28903-Getafe (Madrid), Spain}
\date{}
\maketitle

\begin{abstract}
The main purpose of this paper is to present new families of test statistics
for studying the problem of goodness-of-fit of some data to a latent class
model for binary data. The families of test statistics introduced are based on phi-divergence
measures, a natural extension of maximum likelihood. We also treat the problem of testing a nested sequence of latent class
models for binary data. For these statistics, we obtain their asymptotic distribution.
Finally, a simulation study is carried out in order to compare the efficiency, in the sense of
the level and the power, of the new statistics considered in this paper for sample sizes that are not big enough to apply the asymptotical results.
\end{abstract}

\begin{center}
\bigskip
\end{center}



\section{Introduction}

Consider a set $\mathcal{P}$ of $N$ people: $\mathcal{P}:=\{P_{1},...,P_{N}%
\} $. Each person $P_{v}$ is asked to answers to $k$ dichotomous items $I_{1},...,I_{k};$
let us denote by $y_{vi}$ the answer of person $P_{v}$ to item $I_{i},$ i.e.

\begin{equation*}
y_{vi}:=\left\{
\begin{array}{cl}
1 & \text{if the answer of }P_{v}\text{ to }I_{i}\text{ is correct} \\
0 & \text{otherwise}%
\end{array}%
\right. .
\end{equation*}

Let $\boldsymbol{y}_{v}:=(y_{v1},...,y_{vk})$ denote a generic pattern of
right and wrong answers to the $k$ items given by person $P_{v}.$ In order to
explain the statistical relationships among the observed variables, a
categorical latent variable (categorical unobservable variable) is
postulated to exist, whose different levels partition set $\mathcal{P}$ into
$m$ mutually exclusive and exhaustive {\it latent classes}. Let us denote these
classes by $C_{1},...,C_{m}$ and their corresponding relative sizes by $%
w_{1},...,w_{m};$ thus, $w_{j}$ denotes the probability of a randomly
selected person $P_{v}\in \mathcal{P}$ belongs to class $C_{j},$ i.e.

\begin{equation*}
w_j=Pr(P_{v}\in C_j),\, j=1, ..., m.
\end{equation*}

We denote by $p_{ji}$ the probability of a right answer of $P_{v}$ to the
item $I_i$ under the assumption that $P_{v}$ is in class $C_j:$

\begin{equation*}
p_{ji}= Pr(y_{v i}=1 | P_{v}\in C_j),\, j=1, ..., m,\, i=1, ..., k.
\end{equation*}

Let $\mathbf{y_{\nu }}$ be a possible answer vector. We shall assume that in
each class the answers for the different questions are stochastically
independent; therefore, we can write

\begin{equation*}
Pr(\boldsymbol{y}\mathbf{_{\nu }}|P_{v}\in
C_{j})=\prod_{i=1}^{k}p_{ji}^{y_{\nu i}}(1-p_{ji})^{1-y_{\nu i}},
\end{equation*}%
and
\begin{equation}
Pr(\boldsymbol{y}\mathbf{_{\nu }})=\sum_{j=1}^{m}w_{j}%
\prod_{i=1}^{k}p_{ji}^{y_{\nu i}}(1-p_{ji})^{1-y_{\nu i}}.  \label{eq1}
\end{equation}

There are $2^{k}$ possible answer vectors $\mathbf{y_{\nu }}$ whose
probability of occurrence are given by Eq. \eqref{eq1}; they constitute the
{\it manifest probabilities} for the items $I_{1},...,I_{k}$ in the population
given by $P_{1},...,P_{N}.$ The probability vector $\left( Pr(\boldsymbol{y}%
\mathbf{_{1}}),...,Pr(\boldsymbol{y}\mathbf{_{2^{k}}})\right) $ characterizes
a latent class model (LCM) for binary data.

We will denote by $N_{\nu },\,\nu =1,...,2^{k},$ the number of times that
the sequence $\mathbf{y_{\nu }}$ appears in an $N$-sample and

\begin{equation*}
\widehat{\boldsymbol{p}}:=(N_{1}/N,...,N_{2^{k}}/N).
\end{equation*}

The likelihood function $L$ is given by

\begin{equation}
L(w_{1},...,w_{m},p_{11},...,p_{mk})=Pr(N_{1}=n_{1},...,N_{2^{k}}=n_{2^{k}})=%
{\frac{N!}{\displaystyle \prod_{\nu =1}^{2^{k}}n_{\nu }!}}\prod_{\nu =1}^{2^{k}}\Pr (%
\mathbf{y_{\nu }})^{n_{\nu }}.  \label{eq2}
\end{equation}

By $n_{\nu }$ we are denoting a realization of the random variable $N_{\nu
}, \nu =1, ..., 2^k.$ In this model the unknown parameters are $w_j, j=1,
..., m$ and $p_{ji}, j=1, ..., m, i=1, ..., k.$ These parameters can be
estimated using the maximum likelihood estimator (e.g. McHugh (1956),
Lazarsfeld \& Henry (1968), Clogg (1995)). In order to avoid the problem of
obtaining uninterpretable estimations for the item latent probabilities lying
outside the interval $[0,1],$ some authors (Lazarsfeld \& Henry (1968),
Formann (1976), Formann (1977), Formann (1978), Formann (1982), Formann
(1985)) proposed a linear-logistic parametrization for $w_j $ and $p_{ji}$ given by

\begin{equation*}
p_{ji}={\frac{exp(x_{ji})}{1+exp(x_{ji})}},\,\,j=1,...,m,\,\,i=1,...,k,
\end{equation*}%
and
\begin{equation*}
w_{j}={\frac{exp(z_{j})}{{\displaystyle \sum_{h=1}^{m}exp(z_{h})}}},\,\,j=1,...,m.
\end{equation*}

Next, restrictions are introduced relating parameters $x_{ji}, w_j$ to some
explanatory variables, defined through parameters $\lambda_r, r=1, ..., t$ and $\eta_s, s=1,
..., u$, so the final model is given by

\begin{equation}
p_{ji}={\frac{exp({\displaystyle \sum_{r=1}^{t}q_{jir}\lambda _{r}+c_{ji}})}{1+exp({\displaystyle
\sum_{r=1}^{t}q_{jir}\lambda _{r}+c_{ji}})}},\,\,j=1,...,m,\,\,i=1,...,k,
\label{A}
\end{equation}%
and
\begin{equation}
w_{j}={\frac{exp({\displaystyle \sum_{r=1}^{u}v_{jr}\eta _{r}+d_{j}})}{{\displaystyle
\sum_{h=1}^{m}exp(\sum_{r=1}^{u}v_{hr}\eta _{r}+d_{h})}}},\,\,j=1,...,m,
\label{B}
\end{equation}%
where

\begin{equation*}
\boldsymbol{Q}_{r}=(q_{jir})_{\overset{j=1,...,m}{i=1,...,k}},r=1,...,t,\,\,%
\boldsymbol{C}=(c_{ji})_{\overset{j=1,...,m}{i=1,...,k}},\,\,\boldsymbol{V}%
=(v_{jr})_{\overset{j=1,...,m}{r=1,...,u}},\,\,\boldsymbol{d}%
=(d_{j})_{j=1,...,m},
\end{equation*}%
are fixed. Matrix $\boldsymbol{Q}$ specifies to which amount the predictors
defined through parameters $\lambda _{r}$ are relevant for each $x_{ji}.$
The terms $c_{ji}$ were introduced to include the possibility that certain $%
p_{ji}$ are fixed to certain previously determined values; this possibility
was considered by Goodman (1974). The same applies for matrix $\boldsymbol{V}$:
thus, $\boldsymbol{V}$ specifies to which amount $\eta _{s}$ is relevant for
each $z_{j}.$ The terms $d_{j}$ are introduced to include the possibility
that certain $z_{j}$ are fixed to certain previously determined values.

Consequently, in this case the vector of unknown parameters $\boldsymbol{%
\theta }$ in the LCM for binary data is given by

\begin{equation*}
\boldsymbol{\theta }:=(\boldsymbol{\lambda },\boldsymbol{\eta }),
\end{equation*}%
where $\boldsymbol{\lambda }$ and $\boldsymbol{\eta }$ are defined as

\begin{equation*}
\boldsymbol{\lambda }:=(\lambda _{1},...,\lambda _{t}),\,\boldsymbol{\eta }%
:=(\eta _{1},...,\eta _{u}).
\end{equation*}%

Once $\boldsymbol{\lambda }$ and $\boldsymbol{\eta }$ are estimated, relations
(\ref{A}) and (\ref{B}) give estimations for the parameters $w_{j}$ $%
(j=1,...,m)$ and $p_{ji}$ $(j=1,...,m;$ $i=1,...,k)$

By $\boldsymbol{\Theta }$ we shall denote the set in which the parameter $
\boldsymbol{\theta }$ varies, i.e. the parametric space. Thus, we have $t+u$
unknown parameters that can be estimated by maximum likelihood through Eq. \eqref{eq2}.

In Felipe et al. (2014), a new procedure for estimating $ p_{ji},\,w_{j},\,\,i=1,...,k,\,\,j=1,...,m, $ estimating previously the parameters $\lambda _{i}, i=1,...,t$ and $\eta
_{j}, j=1,...,u$ was presented. It consists in introducing in the
context of LCM for binary data a new family of estimators based on
divergence measures: {\it Minimum $\phi $-divergence estimators} (M$\phi $E). As shown in Felipe at al. (2014), this
family of estimators contains as a particular case the classical maximum
likelihood estimator (MLE). M$\phi $E were introduced for the
first time in Morales et al. (1995) and since then, many interesting
estimation problems have been solved using them (see e.g. Pardo (2006)).

Let us briefly explain this procedure. Consider two probablity distributions ${\bf p}=(p_1, ..., p_M), {\bf q}=(q_1, ..., q_M)$ and a function $\phi $ that is convex for $x>0$ and satisfies $\phi (1)=0,0\phi
(0/0)=0$ and

\begin{equation*}
0\phi (p/0)=p\lim_{x\rightarrow \infty }{\frac{\phi (x)}{x}}.
\end{equation*}

The $\phi $-{\bf divergence measure} between the probability distributions ${\bf p}$ and ${\bf q}$ is defined by

$$ D_{\phi }({\bf p}, {\bf q}):=\sum_{i=1}^M p_i \phi \left( {q_i\over p_i}\right).$$

Given a LCM for binary data with parameters $\boldsymbol{\lambda }=(\lambda _{1},...,\lambda
_{t})$ and $\boldsymbol{\eta }=(\eta _{1},...,\eta _{u}),$ the M$\phi $E of $%
\boldsymbol{\theta }=(\boldsymbol{\lambda },\boldsymbol{\eta })$ is any $%
\hat{\boldsymbol{\theta }}_{\phi }$ satisfying

\begin{equation}
\hat{\boldsymbol{\theta }}_{\phi }=arg\min_{(\boldsymbol{\lambda },%
\boldsymbol{\eta })\in \boldsymbol{\Theta }}D_{\phi }(\widehat{\boldsymbol{p}%
},\boldsymbol{p}(\boldsymbol{\lambda },\boldsymbol{\eta }))  \label{1.1}
\end{equation}
where $D_{\phi }(\widehat{\boldsymbol{p}},\boldsymbol{p}(\boldsymbol{\lambda
},\boldsymbol{\eta }))$ is the $\phi $-divergence measure between the
probability vectors $\widehat{\boldsymbol{p}}$ and $\boldsymbol{p}(%
\boldsymbol{\lambda },\boldsymbol{\eta }),$ given by
\begin{equation}
D_{\phi }(\widehat{\boldsymbol{p}},\boldsymbol{p}(\boldsymbol{\lambda },%
\boldsymbol{\eta }))=\sum_{\nu =1}^{2^{k}}p(\boldsymbol{y}\mathbf{_{\nu }},%
\boldsymbol{\lambda },\boldsymbol{\eta })\phi \left( {\frac{\hat{p}_{\nu }}{%
p(\boldsymbol{y}\mathbf{_{\nu }},\boldsymbol{\lambda },\boldsymbol{\eta )}}}%
\right) .  \label{eq5}
\end{equation}
%

For more details about $\phi $-divergence measures see Cressie and Pardo (2002) and Pardo (2006). In the particular case of $\phi (x)=x\log x-x+1,$ we obtain the so-called {\it Kullback-Leibler divergence measure}, i.e.%
\begin{equation}
D_{Kullback}(\widehat{\boldsymbol{p}},\boldsymbol{p}(\boldsymbol{\lambda },
\boldsymbol{\eta }))=\sum_{\nu =1}^{2^{k}}\hat{p}_{\nu }\log {\frac{\hat{p}
_{\nu }}{p(\boldsymbol{y}\mathbf{_{\nu }},\boldsymbol{\lambda },\boldsymbol{
\eta )}}}.  \label{eq4}
\end{equation}

It is not difficult to establish (see Felipe et al (2014)) that
\begin{equation*}
\log L(w_{1},...,w_{m},p_{11},...,p_{mk})=-ND_{Kullback}(\widehat{%
\boldsymbol{p}},\boldsymbol{p}(\boldsymbol{\lambda },\boldsymbol{\eta }%
))+constant,
\end{equation*}


Therefore, maximizing Eq. (\ref{eq2}) in $\boldsymbol{\lambda }$ and $%
\boldsymbol{\eta }$ is equivalent to minimizing Eq. (\ref{eq4}) in $\boldsymbol{%
\lambda }$ and $\boldsymbol{\eta }$. Consequently, the value $\hat{%
\boldsymbol{\theta }}=(\hat{\boldsymbol{\lambda }},\hat{\boldsymbol{\eta }})$
that minimizes $\boldsymbol{\theta }=(\boldsymbol{\lambda },\boldsymbol{\eta
})$ in the Kullback-Leibler divergence is the MLE of the parameters for the
LCM for binary data or equivalently, the minimum Kullback-Leibler divergence estimator. We
shall denote it by $\hat{\boldsymbol{\theta }}$ or

\begin{equation*}
\hat{\boldsymbol{\theta }}_{Kullback}:=arg\min_{(\boldsymbol{\lambda },%
\boldsymbol{\eta })\in \boldsymbol{\Theta }}D_{Kullback}(\widehat{%
\boldsymbol{p}},\boldsymbol{p}(\boldsymbol{\lambda },\boldsymbol{\eta })).
\end{equation*}
This fact allows us to say that $M\phi E$ is a natural extension of the MLE.

The rest of the paper is organized as follows: In Section 2 we study the problem of goodness-of-fit when dealing with phi-divergence measures; two families of test statistics generalizing the classical ones studied in Formann (1985) are introduced and their asymptotical behavior is established. In Section 3, we proceed the same way for the problem of determining the best model in a nested sequence; besides, we also provide the asymptotical behavior of the test statistics. Section 4 is devoted to a simulation study. We finish with the conclusions. In an appendix we give the proofs of the results presented in Sections 2 and 3.

\section{Goodness-of-fit tests}

LCM for binary data fit is assessed by comparing the observed classification frequencies to
the expected frequencies predicted by the LCM for binary data. When dealing with the MLE, the difference is formally
assessed with a likelihood ratio test statistic or with a chi-square test
statistic whose expressions are given by
\begin{equation}
G^{2}=2N\sum_{\nu =1}^{2^{k}}\hat{p}_{\nu }\log \frac{\hat{p}_{\nu }}{p(%
\boldsymbol{y}\mathbf{_{\nu }},\widehat{\boldsymbol{\lambda }},\widehat{%
\boldsymbol{\eta }})}  \label{2.1}
\end{equation}%
and
\begin{equation}
X^{2}=\sum_{\nu =1}^{2^{k}}\frac{\left( n_{s}-Np(\boldsymbol{y}\mathbf{_{\nu
}},\widehat{\boldsymbol{\lambda }},\widehat{\boldsymbol{\eta }})\right) ^{2}%
}{Np(\boldsymbol{y}\mathbf{_{\nu }},\widehat{\boldsymbol{\lambda }},\widehat{%
\boldsymbol{\eta }})},  \label{2.2}
\end{equation}
respectively.

It is well-known that the asymptotic distribution of the test statistics $%
X^{2}$ and $G^{2}$ is a chi-square distribution with $2^{k}-(u+t)-1$ degrees of
freedom, see Forman (1985). It is a simple exercise to see that these test
statistics are particular cases of the more general
family of test statistics
\begin{equation}
T_{\phi }=\frac{2N}{\phi ^{\prime \prime }(1)}D_{\phi }\left( \widehat{%
\boldsymbol{p}},\boldsymbol{p}(\widehat{\boldsymbol{\lambda }}\boldsymbol{,}%
\widehat{\boldsymbol{\eta }})\right)  \label{2.20}
\end{equation}%
taking $\phi (x)=\frac{1}{2}(x-1)^{2}$ and $\phi (x)=x\log x-x+1,$
respectively. In the following we shall denote $\hat{\boldsymbol{\theta }}%
=\left( \widehat{\boldsymbol{\lambda }}\boldsymbol{,}\widehat{\boldsymbol{%
\eta }}\right) $ and we shall write $D_{\phi }\left( \widehat{\boldsymbol{p}}%
,\boldsymbol{p}(\hat{\boldsymbol{\theta }})\right) .$ Therefore, Eq. \eqref{2.20} gives a family of test statistics
for the problem of goodness-of-it to some data to a LCM. In Eq. \eqref{2.20} parameters $\bm \lambda $ and $\bm \eta $
are estimated using the MLE, but notice that MLE is a particular case of the M$\phi $E.

Based on the M$\phi $E defined in Eq. (\ref{1.1})$,$ we
shall consider in this paper the phi-divergence family of test statistics
given by
\begin{equation}
T_{\phi _{1}}^{\phi _{2}}:=\frac{2N}{\phi _{1}^{\prime \prime }(1)}D_{\phi
_{1}}\left( \widehat{\boldsymbol{p}},\boldsymbol{p}(\hat{\boldsymbol{\theta }%
}_{\phi _{2}})\right)  \label{2.3}
\end{equation}%
where $\hat{\boldsymbol{\theta }}_{\phi _{2}}=($ $\widehat{\boldsymbol{%
\lambda }}_{\phi _{2}},\widehat{\boldsymbol{\eta }}_{\phi _{2}}).$

This family of test statistics is a natural extension of the family (\ref%
{2.20}) in which the MLE has been replaced by the M$\phi _{2}E.$ Notice that in the
family presented in (\ref{2.3}) we have the possibility to use one measure of divergence
based on a function $\phi _{2}$ for the problem estimation and another
measure of divergence based on a function $\phi _{1}$ for the problem of
testing.

In the following theorem we present the asymptotic distribution of this family of test statistics.

\begin{theorem}\label{Th1}
Under the hypothesis that the LCM for binary data with parameters $\boldsymbol{\lambda =}%
(\lambda _{1},...,\lambda _{t})$ and $\boldsymbol{\eta }=(\eta _{1},...,\eta
_{u})$ holds, the asymptotic distribution of the family of test statistics $%
T_{\phi _{1}}^{\phi _{2}}$ given in (\ref{2.3}) is a chi-square
distribution wit $2^{k}-(u+t)-1$ degrees of freedom.
\end{theorem}

\begin{proof}
See Appendix
\end{proof}

It is noteworthy that the asymptotical distribution does not depend on $\phi ,$ i.e. it is the same for any function $\phi $ considered.

Let us see an example:

\begin{example}
We consider the interview data collected by Coleman (1964) and analized later in Goodman (1974); this model is explained in Formann (1982) and Formann (1985). The experiment consists in evaluating the answers of 3398 schoolboys to two questions about their membership in the ``leading crowd" on two occasions $t_1$ and $t_2$ (October, 1957 and May, 1958). Thus, in this model we have 4 questions and there are four manifest variables (answers to both questions at both moments); these answers can only be ``low" (value 0) and ``high" (value 1), so that the manifest variables are dichotomous. The sample data is given in next table:

\begin{center}
\begin{tabular}{c|cccc}
October, 1957/ May, 1958 & 00 & 01 & 10 & 11 \\
\hline 00                & 554 & 338 & 97 & 85 \\
01                       & 281 & 531 & 75 & 184 \\
10                       & 87  & 56  & 182 & 171 \\
11                       & 49   & 110 & 140 & 458 \\
\end{tabular}
\end{center}

Next, 4 latent classes are considered, namely

$C_1\equiv $ low agreement in question 1 and low agreement in question 2.

$C_2\equiv $ low agreement in question 1 and high agreement in question 2.

$C_3\equiv $ high agreement in question 1 and low agreement in question 2.

$C_4\equiv $ high agreement in question 1 and high agreement in question 2.

There are 16 probability values $p_{ji}$ to be estimated; we consider the first hypothesis appearing in Formann (1985), namely ``The attitudinal changes between times $t_1$ and $t_2$ are dependent on the positions (low, high) of the respective classes on the underlying attitudinal scales at $t_1$". Thus, a model with 8 parameters $\lambda_i$ is considered; $\lambda_1$ means low agreement in the first question at time $t_1$, $\lambda_2$ means high agreement in the first question at time $t_1$, $\lambda_3$ means low agreement in the second question at time $t_1$, $\lambda_4$ means high agreement in the second question at time $t_1$, and $\lambda_5, \lambda_6, \lambda_7, \lambda_8$ are the same parameters at time $t_2.$ We write the values for matrices ${\bf Q}_i$ as they appear in Formann (1985). In our notation, the matrices ${\bf Q}_i$ can be derived considering the $i$-th column in the table and dividing it in four columns of four elements each (each corresponding to a latent class).

\begin{center}
\begin{tabular}{c|c|cccccccc}
Class & Item & $\lambda_1$ &  $\lambda_2$ & $\lambda_3$ & $\lambda_4$ & $\lambda_5$ & $\lambda_6$ & $\lambda_7$ & $\lambda_8$ \\ \hline 1 & 1 & 1 & 0 & 0 & 0 & 0 & 0 & 0 & 0 \\ & 2 & 0 & 0 & 1 & 0 & 0 & 0 & 0 & 0 \\ & 3 & 0 & 0 & 0 & 0 & 1 & 0 & 0 & 0 \\ & 4 & 0 & 0 & 0 & 0 & 0 & 0 & 1 & 0 \\ \hline 2 & 1 & 1 & 0 & 0 & 0 & 0 & 0 & 0 & 0 \\ & 2 & 0 & 0 & 0 & 1 & 0 & 0 & 0 & 0 \\ & 3 & 0 & 0 & 0 & 0 & 1 & 0 & 0 & 0 \\ & 4 & 0 & 0 & 0 & 0 & 0 & 0 & 0 & 1 \\ \hline 3 & 1 & 0 & 1 & 0 & 0 & 0 & 0 & 0 & 0 \\ & 2 & 0 & 0 & 1 & 0 & 0 & 0 & 0 & 0 \\ & 3 & 0 & 0 & 0 & 0 & 0 & 1 & 0 & 0 \\ & 4 & 0 & 0 & 0 & 0 & 0 & 0 & 1 & 0 \\ \hline 4 & 1 & 0 & 1 & 0 & 0 & 0 & 0 & 0 & 0 \\ & 2 & 0 & 0 & 0 & 1 & 0 & 0 & 0 & 0 \\ & 3 & 0 & 0 & 0 & 0 & 0 & 1 & 0 & 0 \\ & 4 & 0 & 0 & 0 & 0 & 0 & 0 & 0 & 1 \\ \hline
\end{tabular}
\end{center}

Note that the hypothesis is that the attitudinal changes between times $t_1$ and $t_2$ are dependent upon the items as well as on the classes. For this reason, the part corresponding each latent class can be partitioned in four submatrices of size 2$\times $4. The submatrices lying on the main diagonal are the same by the hypothesis defining the model and the two other submatrices are null. The differences among them are due to the differences in the latent classes. Next, $c_{ij}=0,\, \, \forall i, j$ (as we have explained when values $c_{ij}$ were introduced in Section 1). Finally, 4 parameters $\eta_j$ are considered, taking as matrix ${\mathbf V}$ the identity matrix and $d_j=0,\, \forall j.$

It is noteworthy that our model assumes that answers to the questions are
conditionally independent given the latent class. In this example, we are dealing with repeated
responses to two questions, so this assumption may be unrealistic. However, this assumption is made in the original paper of Goodman (1974) and we follow this assumption for the sake of the example.

In order to study if the data are from a LCM for binary data we shall consider the
particular family of phi-divergence measures introduced and studied by
Cressie and Read (1984): The power divergence family. This family is obtained from
\begin{equation}
\phi (x)\equiv \phi _{a}(x)=\left\{
\begin{array}{cl}
{\frac{1}{a(a+1)}}(x^{a+1}-x-a(x-1)) & a\neq 0,a\neq -1 \\
x\log x-x+1 & a=0 \\
-\log x+x-1 & a=-1%
\end{array}%
\right.   \label{eq9}
\end{equation}%
In Felipe et al. (2014) it was established, on the basis of a simulation study,
that a good alternative to the MLE is the M$\phi $E obtained from Eq. (\ref{1.1}) with $a=2/3$, i.e.,
\begin{equation*}
\hat{\boldsymbol{\theta }}_{2/3}=arg\min_{(\boldsymbol{\lambda },\boldsymbol{%
\eta })\in \boldsymbol{\Theta }}D_{2/3}(\hat{\boldsymbol{p}},\boldsymbol{p}(%
\boldsymbol{\lambda },\boldsymbol{\eta })),
\end{equation*}%
being
\begin{equation*}
D_{2/3}(\hat{\boldsymbol{p}},\boldsymbol{p}(\boldsymbol{\lambda },%
\boldsymbol{\eta }))={\frac{9}{10}}\left( {\sum_{j=1}^{2^{k}}\frac{\hat{p}%
_{j}^{5/3}}{p_{j}(\boldsymbol{\lambda },\boldsymbol{\eta })^{2/3}}}-1\right)
\end{equation*}%
Therefore we are going to consider in our study the M$\phi $E obtained with $%
\phi \left( x\right) $ defined in Eq. (\ref{eq9}) for $a=2/3$ in order to get
an estimation of parameters $\boldsymbol{\lambda }$ and $\boldsymbol{\eta .}$
In Table 1 we present the values obtained for these parameters, as well as the estimation of the probabilities and the weights of the latent classes

\begin{table}[h]
\begin{center}
\begin{tabular}{|c|r|r|c|}
\hline
Parameter / a &  & Parameter / a &  \\ \hline
$\hat{\lambda}_{1}$ & -2.34292610 & \multicolumn{1}{|c|}{$\hat{p}_{1,1}$} &
0.08762969 \\
$\hat{\lambda}_{2}$ & 1.72393168 & \multicolumn{1}{|c|}{$\hat{p}_{1,2}$} &
0.30144933 \\
$\hat{\lambda}_{3}$ & -0.84040580 & \multicolumn{1}{|c|}{$\hat{p}_{1,3}$} &
0.11256540 \\
$\hat{\lambda}_{4}$ & 1.56524945 & \multicolumn{1}{|c|}{$\hat{p}_{1,4}$} &
0.28671773 \\
$\hat{\lambda}_{5}$ & -2.06480043 & \multicolumn{1}{|c|}{$\hat{p}_{2,1}$} &
0.08762969 \\
$\hat{\lambda}_{6}$ & 2.29928080 & \multicolumn{1}{|c|}{$\hat{p}_{2,2}$} &
0.82710532 \\
$\hat{\lambda}_{7}$ & -0.91137901 & \multicolumn{1}{|c|}{$\hat{p}_{2,3}$} &
0.11256540 \\
$\hat{\lambda}_{8}$ & 2.01252338 & \multicolumn{1}{|c|}{$\hat{p}_{2,4}$} &
0.88210569 \\
$\hat{\eta}_{1}$ & 0.50480183 & \multicolumn{1}{|c|}{$\hat{p}_{3,1}$} &
0.8463457 \\
$\hat{\eta}_{2}$ & 0.16964329 & \multicolumn{1}{|c|}{$\hat{p}_{3,2}$} &
0.30144933 \\
$\hat{\eta}_{3}$ & -0.87356633 & \multicolumn{1}{|c|}{$\hat{p}_{3,3}$} &
0.90881746 \\
$\hat{\eta}_{4}$ & -0.00424661 & \multicolumn{1}{|c|}{$\hat{p}_{3,4}$} &
0.28671773 \\
$\hat{w}_{1}$ & 0.38936544 & \multicolumn{1}{|c|}{$\hat{p}_{4,1}$} &
0.84863457 \\
$\hat{w}_{2}$ & 0.27848377 & \multicolumn{1}{|c|}{$\hat{p}_{4,2}$} &
0.82710532 \\
$\hat{w}_{3}$ & 0.09811597 & \multicolumn{1}{|c|}{$\hat{p}_{4,3}$} &
0.90881746 \\
$\hat{w}_{4}$ & 0.23403482 & \multicolumn{1}{|c|}{$\hat{p}_{4,4}$} &
0.88210569 \\ \hline
\end{tabular}%
\caption{Estimations of the parameters for Example 4.}
\end{center}
\end{table}

Now we are interested in studying the goodness-of-fit of our data to this model.
We shall consider the family of test statistics, $T_{\phi _{a}}^{\phi
_{2/3}},$ obtained from $\phi _{a}(x)$ with $a=-1,-1/2,$ $0,$ $2/3,$ $%
1,1.5,2,2.5$ and $3,$i.e$.,$
\begin{equation}\label{pag7}
T_{\phi _{a}}^{\phi _{2/3}}=\left\{
\begin{array}{cc}
\frac{2N}{a(a+1)}\left( \sum\limits_{\nu =1}^{2^{k}}\frac{\widehat{p}%
_{\upsilon }^{a+1}}{p\left( \boldsymbol{y}_{\nu },\hat{\boldsymbol{\theta }}%
_{2/3}\right) ^{a}}-1\right)  & \text{if }a\neq 0,-1 \\
2\sum\limits_{\nu =1}^{2^{k}}n_{\upsilon }\log \frac{\widehat{p}_{\upsilon }%
}{p\left( \boldsymbol{y}_{\nu },\hat{\boldsymbol{\theta }}_{2/3}\right) } &
a=0 \\
2N\sum\limits_{\nu =1}^{2^{k}}p\left( \boldsymbol{y}_{\nu },\hat{\boldsymbol{%
\theta }}_{2/3}\right) \log \frac{p\left( \boldsymbol{y}_{\nu },\hat{%
\boldsymbol{\theta }}_{2/3}\right) }{\widehat{p}_{\upsilon }} & a=-1%
\end{array}%
\right. .
\end{equation}%

The results are presented in the following table%

\begin{equation*}
\begin{tabular}{l|ccccccccc}
a & -1 & -1/2 & 0 & 2/3 & 1 & 3/2 & 2 & 5/2 & 3 \\ \hline
$T_{\phi _{a}}^{\phi _{2/3}}$ & 1.279 & 1.278 & 1.277 & 1.277 & 1.277 &
1.277 & 1.278 & 1.279 & 1.281%
\end{tabular}%
\end{equation*}%

On the other hand, the distribution of this statistics is a $\chi^2$ with 16-11-1=4 degrees of freedom; as $\chi _{4; 0.05}^{2}=$ 9.49, we conclude that we have no  evidence to reject our model.

Notice that the values for all test statistics are very similar; this was expected, as the sample size under consideration is big enough ($N=3398$) to apply the asymptotical result of Theorem 1.
\end{example}

\begin{remark}
There are some classical measures of divergence which cannot be expressed as a $\phi $%
-divergence measure, such as the divergence measures of Bhattacharya (1943), R%
\'{e}nyi (1961), and Sharma and Mittal (1977). However, such measures are particular cases of the $(h,\phi )$-divergence measures and can be
defined by
\[
D_{\phi_1 }^{h}\left( \widehat{\boldsymbol{p}},\boldsymbol{p}\left( \widehat{%
\boldsymbol{\lambda }}_{\phi_2}\boldsymbol{,}\widehat{\boldsymbol{\eta }}_{\phi_2}\right)
\right) :=h\left( D_{\phi_1 }\left( \widehat{\boldsymbol{p}},\boldsymbol{p}%
\left( \widehat{\boldsymbol{\lambda }}_{\phi_2}\boldsymbol{,}\widehat{\boldsymbol{%
\eta }}_{\phi_2}\right) \right) \right) ,
\]
where $h$ is a differentiable increasing function mapping from $\left[
0,\infty \right) $ onto $\left[ 0,\infty \right) $, with $h(0)=0$ and $
h^{\prime }(0)>0$. In Table \ref{t1}, these divergence
measures are presented, along with the corresponding expressions of $h$ and $\phi $.

\begin{table}[htbp] \tabcolsep0.8pt  \centering%
$%
\begin{tabular}{ccccc}
\hline
Divergence & \hspace*{0.5cm} & $h\left( x\right) $ & \hspace*{0.5cm} & $\phi
\left( x\right) $ \\ \hline
\multicolumn{1}{l}{R\'{e}nyi} &  & \multicolumn{1}{l}{$\frac{1}{a\left(
a-1\right) }\log \left( a\left( a-1\right) x+1\right) ,\quad a\neq 0,1$} &
& \multicolumn{1}{l}{$\frac{x^{a}-a\left( x-1\right) -1}{a\left( a-1\right) }%
,\quad a\neq 0,1$} \\
\multicolumn{1}{l}{Sharma-Mittal} &  & \multicolumn{1}{l}{$\frac{1}{b-1}%
\left\{ [1+a\left( a-1\right) x]^{\frac{b-1}{a-1}}-1\right\} ,\quad b,a\neq
1 $} &  & \multicolumn{1}{l}{$\frac{x^{a}-a\left( x-1\right) -1}{a\left(
a-1\right) },\quad a\neq 0,1$} \\
\multicolumn{1}{l}{Battacharya} &  & \multicolumn{1}{l}{$-\log \left(
-x+1\right) $} &  & \multicolumn{1}{l}{$-x^{1/2}+\frac{1}{2}\left(
x+1\right) $} \\ \hline
\end{tabular}%
\ \ \ \ \ \ \ \ \ \ \ \ \ \ \ $%
\caption{Some specific $(h,\phi)$-divergence
measures.\label{t1}}%
\end{table}%
The $(h,\phi )$-divergence measures were introduced in Men\'{e}ndez et al.
(1995) and some associated asymptotic results for them were established in Men\'{e}%
ndez et al. (1997). Moreover, some interesting results about R\'{e}nyi divergence
measures can be seen in Gil et al. (2013), Golshani et al. (2009, 2010) and
Nadarajah and Zografos (2003).
\end{remark}

If we deal with $(h, \phi )$-divergence measures in our context, the following can be proved:

\begin{theorem}
Under the assumptions of Theorem 1, the asymptotic distribution of the
family of empirical test statistics defined by
\[
S^{\phi_1 ,h}\left( \widehat{\boldsymbol{p}},\boldsymbol{p}\left( \widehat{%
\boldsymbol{\lambda }}_{\phi_2}\boldsymbol{,}\widehat{\boldsymbol{\eta }}_{\phi_2}\right)
\right) :=\frac{2N}{\phi ^{\prime \prime }_1(1)h^{\prime }(0)}h\left( D_{\phi_1
}\left( \widehat{\boldsymbol{p}},\boldsymbol{p}\left( \widehat{\boldsymbol{%
\lambda }}_{\phi_2}\boldsymbol{,}\widehat{\boldsymbol{\eta }}_{\phi_2}\right) \right) \right)
\]%
is chi-square with $2^{k}-(u+t)-1$ degrees of freedom.
\end{theorem}

\begin{proof}
We have%
\[
h(x)=h(0)+h^{\prime }(0)x+o(x)= h^{\prime }(0)x+o(x),
\]%
and so
\[
h\left( D_{\phi_1 }\left( \left( \widehat{\boldsymbol{p}},\boldsymbol{p}\left(
\widehat{\boldsymbol{\lambda }}_{\phi_2}\boldsymbol{,}\widehat{\boldsymbol{\eta }}_{\phi_2}%
\right) \right) \right) \right) =h^{\prime }(0)D_{\phi_1 }\left( \widehat{%
\boldsymbol{p}},\boldsymbol{p}\left( \widehat{\boldsymbol{\lambda }}_{\phi_2}%
\boldsymbol{,}\widehat{\boldsymbol{\eta }}_{\phi_2}\right) \right) +o\left( D_{\phi_1}\left( \widehat{\boldsymbol{p}},\boldsymbol{p}\left( \widehat{\boldsymbol{%
\lambda }}_{\phi_2}\boldsymbol{,}\widehat{\boldsymbol{\eta }}_{\phi_2}\right) \right) \right)
.
\]%
Then, the required result follows upon applying Theorem \ref{Th1}.
\end{proof}

\section{Nested latent class models}

In the previous example we have seen that the LCM proposed (that we will call $M_{1}$) fits our data; however, a question arises: Is it possible to find a
latent model with a reduced number of parameters that also fits the data? If the
answer is positive, the reduced model should be used instead of $M_1$.

\begin{example}
Consider the example studied in the previous section. In Formann (1985), the following reduced models are studied:

$M_{2}:$ Attitudinal changes between the two moments are dependent on the
latent classes but are independent on the items.


$M_{3}$: Attitudinal changes between the two moments are independent both on
the items and on the latent classes.

$M_{4}:$ There are no attitudinal changes.

These different models imply different number of parameters $\lambda _{i}.$
More concretely, model $M_{2}$ needs six parameters $\lambda _{i}
$, model $M_{3}$ needs five parameters and finally model $M_{4}$ needs four
parameters. The corresponding matrices $Q_{i}$ for these models can be found in Formann
(1985).

As for $M_1$, $c_{ij}=0,\,\,\forall i,j$ and 4 parameters $\eta _{j}$ are
considered, taking matrix $V$ as the identity matrix and $d_{j}=0,\,\forall j.$

We can observe that
\begin{equation*}
\bm \Theta _{M_{1}}\supset \bm \Theta _{M_{2}}\supset \bm \Theta _{M_{3}}\supset \bm \Theta
_{M_{4}},
\end{equation*}%
being $\bm \Theta _{M_{i}}$ the parameter space associated to the LCM $%
M_{i}.$
Therefore, we have a nested sequence of LCM.
\end{example}

In general, we shall assume that we have $m$ LCM $\left\{ M_{l}\right\}
_{l=1,...,m}$ in such a way that the parameter space associated to $M_{l},$ $%
l=1,...,m,$ is $\bm \Theta _{M_{l}}$ and
\begin{equation*}
\bm \Theta _{M_{m}}\subset \bm \Theta _{M_{m-1}}\subset ....\subset \bm \Theta
_{M_{1}}\subset \mathbb{R}^{t}
\end{equation*}%
holds. Let us denote $\dim \left( \bm \Theta _{M_{l}}\right) =h_{l};$ $l=1,....,m$ with
\begin{equation*}
h_{m}<h_{m-1}<....<h_{1}\leq t,
\end{equation*}%
i.e., the parameters of one LCM are a subset of the
parameters of the other.
Our strategy is to test successively
\begin{equation}
H_{l+1}:\boldsymbol{\theta \in }\bm \Theta _{M_{l+1}}\text{ against }H_{l}:%
\boldsymbol{\theta \in }\bm \Theta _{M_{l}},\text{ }l=1,...,m-1.  \label{Nested}
\end{equation}%

We continue to test as long as the null hypothesis is accepted, and choose
the LCM $M_{l}$ with parameter space $\bm \Theta _{M_{l}}$ according to
the first $l$ satisfying that $H_{l+1}$ is rejected (as null hypothesis) in favor
of $H_{l}$ (as alternative hypothesis). This strategy is quite standard
for nested models (Cressie et al., 2003). In this section we present two families of phi-divergence test statistics for solving the tests presented in \eqref{Nested}.


%
%
%

Let us introduce some additional notation in order to be able to
formulate the nested LCM in a convenient way for our purposes. We shall denote by $\boldsymbol{\theta }^{A}=\left( \boldsymbol{\theta }%
^{A,1},\boldsymbol{\theta }^{A,2},\boldsymbol{\theta }^{A,3},\boldsymbol{%
\theta }^{A,4}\right) $ with $\boldsymbol{\theta }^{A,1}=\left( \lambda
_{1},...,\lambda _{t^{\ast }}\right) ,$ $\boldsymbol{\theta }^{A,2}=\left(
\lambda _{t^{\ast }+1},...,\lambda _{t}\right) ,$ $\boldsymbol{\theta }%
^{A,3}=\left( \eta _{1},...,\eta _{u^{\ast }}\right) $ and $\boldsymbol{%
\theta }^{A,4}=$ $\left( \eta _{u^{\ast }+1},...,\eta _{u}\right) $ the
parameters associated to the LCM $A$ and by $\boldsymbol{\theta }%
^{B}=\left( \boldsymbol{\theta }^{A,1},\boldsymbol{0},\boldsymbol{\theta }%
^{A,3},\boldsymbol{0}\right) $ the parameters associated to the LCM $B.$ We
shall assume that $t+u=h_{1}$ and $t^{\ast }+u^{\ast }=h_{2}$ . It is clear
that the LCM $B$ is nested in LCM $A.$

It can be observed that the testing problem given in \eqref{Nested} can be  equivalently formulated using the previous notation in the following way:

\begin{equation}
H_{Null}:\boldsymbol{\theta }^{A,2}=\boldsymbol{0}_{t-t^{\ast }}\text{ and }%
\boldsymbol{\theta }^{A,4}=\boldsymbol{0}_{u-u^{\ast }}.  \label{4.0}
\end{equation}%

The expression of the classical likelihood ratio test for solving \eqref{4.0}
is
\begin{equation}
G_{A-B}^{2}=2\sum_{\nu =1}^{2^{k}}n_{\upsilon }\log \frac{p\left(
\boldsymbol{y}_{\nu },\widehat{\boldsymbol{\theta }}^{A}\right) }{p\left(
\boldsymbol{y}_{\nu },\widehat{\boldsymbol{\theta }}^{B}\right) }.
\label{4.1}
\end{equation}%

Notice that not only the likelihood ratio test can be used for testing \eqref{4.0}; the chi-square test statistic given by
\begin{equation}
X_{A-B}^{2}=N\sum_{\nu =1}^{2^{k}}\frac{\left( p\left( \boldsymbol{y}_{\nu },%
\widehat{\boldsymbol{\theta }}^{A}\right) -p\left( \boldsymbol{y}_{\nu },%
\widehat{\boldsymbol{\theta }}^{B}\right) \right) ^{2}}{p\left( \boldsymbol{y%
}_{\nu },\widehat{\boldsymbol{\theta }}^{B}\right) }  \label{4.2}
\end{equation}%
can be also used instead.

We can observe that
\begin{equation}
G_{A-B}^{2}=2N\left( D_{Kullback}\left( \hat{\boldsymbol{p}},\boldsymbol{p}(%
\widehat{\boldsymbol{\theta }}^{A})\right) -D_{Kullback}\left( \hat{%
\boldsymbol{p}},\boldsymbol{p}(\widehat{\boldsymbol{\theta }}^{B})\right)
\right)  \label{4.3}
\end{equation}%
and
\begin{equation}
X_{A-B}^{2}=\frac{2N}{\phi '' \left( 1\right) }D_{Pearson}(\boldsymbol{p}(\widehat{\boldsymbol{\theta }}%
^{A}),\boldsymbol{p}(\widehat{\boldsymbol{\theta }}^{B}))  \label{4.4}
\end{equation}%
being $D_{Pearson}(\boldsymbol{p}(\widehat{\boldsymbol{\theta }}^{A}),%
\boldsymbol{p}(\widehat{\boldsymbol{\theta }}^{B}))$ the phi-divergence
measure between the probability vectors $\boldsymbol{p}(\widehat{\boldsymbol{%
\theta }}^{A})$ and $\boldsymbol{p}(\widehat{\boldsymbol{\theta }}^{B})$
with
\begin{equation*}
\phi \left( x\right) =\frac{1}{2}\left( x-1\right) ^{2}.
\end{equation*}%
Based on Eqs. (\ref{4.3}) and (\ref{4.4}) we are going to give two families of test statistics that are
natural extensions of these test statistics for solving the problem of testing given in \eqref{Nested}.

A generalization of (\ref{4.3}) is obtained if we replace the
Kullback-Leibler divergence measure for a phi-divergence measure, i.e.,
\begin{equation}
S_{A-B}^{\phi _{1,}\phi _{2}}=\frac{2N}{\phi _{1}^{\prime \prime }\left(
1\right) }\left( D_{\phi _{1}}\left( \hat{\boldsymbol{p}},\boldsymbol{p}(%
\widehat{\boldsymbol{\theta }}_{\phi _{2}}^{A})\right) -D_{\phi _{1}}\left(
\hat{\boldsymbol{p}},\boldsymbol{p}(\widehat{\boldsymbol{\theta }}_{\phi
_{2}}^{B})\right) \right)  ,\label{4.5}
\end{equation}%
and a generalization of (\ref{4.4}) is achieved if we replace the Pearson divergence
measure for a phi-divergence measure, i.e.,
\begin{equation}
T_{A-B}^{\phi _{1,}\phi _{2}}=\frac{2N}{\phi _{1}^{\prime \prime }\left(
1\right) }D_{\phi _{1}}\left( \boldsymbol{p}(\widehat{\boldsymbol{\theta }}%
_{\phi _{2}}^{A}),\boldsymbol{p}(\widehat{\boldsymbol{\theta }}_{\phi
_{2}}^{B})\right) .  \label{4.6}
\end{equation}%

The previous extensions have been considered in many statistical applications, see for example Cressie et al. (2003), Pardo (2006) and references therein.

In the following theorem we shall obtain the asymptotic distribution of the
family of test statistics given in (\ref{4.5}) and (\ref{4.6}).

\begin{theorem}
Given the LCM for binary data $A$ and $B$ with parameters $\boldsymbol{\theta }^{A}=\left(
\boldsymbol{\theta }^{A,1},\boldsymbol{\theta }^{A,2},\boldsymbol{\theta }%
^{A,3},\boldsymbol{\theta }^{A,4}\right) $ and $\boldsymbol{\theta }%
^{B}=\left( \boldsymbol{\theta }^{A,1},\boldsymbol{0},\boldsymbol{\theta }%
^{A,3},\boldsymbol{0}\right) $, respectively, and under the null hypothesis
given in (\ref{4.0}), it follows
\begin{equation*}\label{4.5}
S_{A-B}^{\phi _{1,}\phi _{2}}\overset{\mathcal{L}}{\underset{%
N\longrightarrow \infty }{\mathcal{\rightarrow }}}\chi _{h_{1}-h_{2}}^{2},
\end{equation*}%
and
\begin{equation*}\label{4.6}
T_{A-B}^{\phi _{1,}\phi _{2}}\overset{\mathcal{L}}{\underset{%
N\longrightarrow \infty }{\mathcal{\rightarrow }}}\chi _{h_{1}-h_{2}}^{2}.
\end{equation*}
\end{theorem}

\begin{proof}
See Appendix.
\end{proof}

\begin{example}
(Continuation of Example 1) We shall consider the sequence of LCM
\begin{equation*}
\bm \Theta _{M_{1}}\supset \bm \Theta _{M_{2}}\supset \bm \Theta _{M_{3}}\supset \bm \Theta
_{M_{4}},
\end{equation*}%
In a similar way as in
the previous section we consider $\hat{\boldsymbol{\theta }}_{2/3}$ in order
to estimate the parameters of the different models. For testing, we consider the
family of phi-divergences test statistics  $S_{A-B}^{\phi _{1,}\phi _{2}}$ and $T_{A-B}^{\phi_1, \phi_2}$
given in (\ref{4.5}) and (\ref{4.6}), being $\phi _{2}(x)=\phi _{a}(x),$ with $a=2/3,$ and $%
\phi _{a}(x)$ defined in (\ref{eq9}). For $\phi _{1}(x)$ we shall take $\phi
_{a}(x)$ with $a=-1,-1/2,0,2/3,3/2,2,5/2$ and $3$. In Table 3 we
present the results obtained.

\begin{table}[h]
\begin{center}
\begin{tabular}{|c|ccc|c|ccc|}
\hline a/Model & $M_1-M_2$ & $M_2-M_3$ & $M_3-M_4$ & & $M_1-M_2$ & $M_2-M_3$ & $M_3-M_4$ \\
\hline -1      & 3.761     & 4.610     & 31.465 & & 3.431 & 4.613 & 31.005 \\
      -1/2     & 3.757     & 4.593     & 30.977 & & 3.417 & 4.604 & 30.845\\
       0       & 3.755     & 4.584     & 30.769 & & 3.403 & 4.595 & 30.722\\
       2/3     & 3.754     & 4.578     & 30.626 & & 3.386 & 4.585 & 30.616\\
       1       & 3.754     & 4.580     & 30.659 & & 3.378 & 4.580 & 30.587\\
       3/2     & 3.756     & 4.586     & 30.820 & & 3.366 & 4.574 & 30.574\\
       2       & 3.759     & 4.599     & 30.991 & & 3.355 & 4.570 & 30.597\\
       5/2     & 3.763     & 4.617     & 31.347 & & 3.344 & 4.566 & 30.655\\
       3       & 3.769     & 4.641     & 31.765 & & 3.334 & 4.563 & 30.749\\
\hline $\chi^2_{i;0.05}$ & 5.99   & 3.84      & 3.84 & & 5.99   & 3.84      & 3.84\\
\hline
\end{tabular}
\end{center}
\caption{Results for Example 8 for statistics $S$ (left) and $T$ (right).}
\end{table}

As a conclusion, we can adopt LCM $M_2$ as the best model in all cases. As before, the values obtained are very similar, due to the asymptotical results.
\end{example}

\begin{remark}
Using the ideas given in Remark 3 we can consider the
following two families of $(h,\phi )$-divergence test statistics:%
\[
S_{A-B}^{\phi _{1,}\phi _{2},h}=\frac{2N}{\phi _{1}^{\prime \prime }\left(
1\right) h^{\prime }(0)}\left( h\left( D_{\phi _{1}}\left( \hat{\boldsymbol{p%
}},\boldsymbol{p}(\widehat{\boldsymbol{\theta }}_{\phi _{2}}^{A})\right)
\right) -h\left( D_{\phi _{1}}\left( \hat{\boldsymbol{p}},\boldsymbol{p}(%
\widehat{\boldsymbol{\theta }}_{\phi _{2}}^{B})\right) \right) \right) ,
\]%
and
\[
T_{A-B}^{\phi _{1},\phi _{2},h}=\frac{2N}{\phi _{1}^{\prime \prime
}(1)h^{\prime }(0)}h\left( D_{\phi }\left( \boldsymbol{p}(\widehat{%
\boldsymbol{\theta }}_{\phi _{2}}^{A}),\boldsymbol{p}(\widehat{\boldsymbol{%
\theta }}_{\phi _{2}}^{B})\right) \right) .
\]%
It is easy to establish that again
\[
S_{A-B}^{\phi _{1,}\phi _{2},h}\overset{\mathcal{L}}{\underset{%
N\longrightarrow \infty }{\mathcal{\rightarrow }}}\chi _{h_{1}-h_{2}}^{2},
\]
and
\[
T_{A-B}^{\phi _{1},\phi _{2},h}\overset{\mathcal{L}}{\underset{%
N\longrightarrow \infty }{\mathcal{\rightarrow }}}\chi _{h_{1}-h_{2}}^{2}.
\]
\end{remark}

\section{Simulation}

Sections 2 and 3 present theoretical results for testing hypothesis in latent models with binary data. These results give the asymptotic distribution theory for the phi-divergence test statistics given in \eqref{2.3}, \eqref{4.5} and \eqref{4.6} under the null hypothesis. In this section we present a simulation study to analyze the behavior of this statistics in small samples. We shall analyze the test statistics given in \eqref{2.3}.


In Felipe et al. (2014), it was established that the best way to estimate the unknown parameters from the point of view of the efficiency as well as the robustness was the minimum power divergence obtained for $a=2/3$, as this estimator balances infinitesimal robustness and asymptotic efficiency.

%

Therefore, in our simulation study we shall consider this estimator. We compare the different test statistics of the family $T_{\Phi_a}^{\Phi_{2/3}}$ defined in \eqref{pag7}. The theoretical LCM with binary data that we shall consider in our simulation study is given by a theoretical model with 5 dichotomous questions and 10 latent classes; next, 7 parameters $\lambda_j$ and 6 parameters $\eta_k$ are considered; the corresponding matrices of the model are

{\small
$$ {\bf Q}_1= \left( \begin{array}{ccccc}
                         1 & 0 & 0 & 0 & 0 \\
                         0 & 0 & 0 & 0 & 0 \\
                         0 & 0 & 0 & 0 & 0 \\
                         0 & 0 & 0 & 0 & 1 \\
                         0 & 0 & 0 & 1 & 0 \\
                         0 & 0 & 1 & 0 & 0 \\
                         0 & 1 & 0 & 0 & 0 \\
                         1 & 0 & 0 & 0 & 0 \\
                         0 & 0 & 0 & 1 & 0 \\
                         0 & 0 & 0 & 0 & 1 \\
\end{array}\right) ,
{\bf Q}_2 =  \left( \begin{array}{ccccc}
                         0 & 1 & 0 & 0 & 0 \\
                         1 & 0 & 0 & 0 & 0 \\
                         0 & 0 & 0 & 0 & 0 \\
                         0 & 0 & 0 & 0 & 0 \\
                         0 & 0 & 0 & 0 & 1 \\
                         0 & 0 & 0 & 1 & 0 \\
                         0 & 0 & 1 & 0 & 0 \\
                         0 & 0 & 0 & 0 & 1 \\
                         1 & 0 & 0 & 0 & 0 \\
                         0 & 0 & 1 & 0 & 0 \\
\end{array}\right) , {\bf Q}_3= \left( \begin{array}{ccccc}
                         0 & 0 & 1 & 0 & 0 \\
                         0 & 1 & 0 & 0 & 0 \\
                         1 & 0 & 0 & 0 & 0 \\
                         0 & 0 & 0 & 0 & 0 \\
                         0 & 0 & 0 & 0 & 0 \\
                         0 & 0 & 0 & 0 & 1 \\
                         0 & 0 & 0 & 1 & 0 \\
                         0 & 1 & 0 & 0 & 0 \\
                         0 & 0 & 0 & 0 & 1 \\
                         1 & 0 & 0 & 0 & 0 \\
\end{array}\right) , {\bf Q}_4=\left( \begin{array}{ccccc}
                         0 & 0 & 0 & 1 & 0 \\
                         0 & 0 & 1 & 0 & 0 \\
                         0 & 1 & 0 & 0 & 0 \\
                         1 & 0 & 0 & 0 & 0 \\
                         0 & 0 & 0 & 0 & 0 \\
                         0 & 0 & 0 & 0 & 0 \\
                         0 & 0 & 0 & 0 & 1 \\
                         0 & 0 & 0 & 0 & 0 \\
                         0 & 1 & 0 & 0 & 0 \\
                         0 & 0 & 0 & 0 & 0 \\
\end{array}\right) $$ }

$$ {\bf Q}_5=\left( \begin{array}{ccccc}
                         0 & 0 & 0 & 0 & 1 \\
                         0 & 0 & 0 & 1 & 0 \\
                         0 & 0 & 1 & 0 & 0 \\
                         0 & 1 & 0 & 0 & 0 \\
                         1 & 0 & 0 & 0 & 0 \\
                         0 & 0 & 0 & 0 & 0 \\
                         0 & 0 & 0 & 0 & 0 \\
                         0 & 0 & 1 & 0 & 0 \\
                         0 & 0 & 0 & 0 & 0 \\
                         0 & 0 & 0 & 1 & 0 \\
\end{array}\right) ,\, \, {\bf Q}_6= \left( \begin{array}{ccccc}
                         0 & 0 & 0 & 0 & 0 \\
                         0 & 0 & 0 & 0 & 1 \\
                         0 & 0 & 0 & 1 & 0 \\
                         0 & 0 & 1 & 0 & 0 \\
                         0 & 1 & 0 & 0 & 0 \\
                         1 & 0 & 0 & 0 & 0 \\
                         0 & 0 & 0 & 0 & 0 \\
                         0 & 0 & 0 & 0 & 0 \\
                         0 & 0 & 1 & 0 & 0 \\
                         0 & 1 & 0 & 0 & 0 \\
\end{array}\right) ,\, \, {\bf Q}_7=\left( \begin{array}{ccccc}
                         0 & 0 & 0 & 0 & 0 \\
                         0 & 0 & 0 & 0 & 0 \\
                         0 & 0 & 0 & 0 & 1 \\
                         0 & 0 & 0 & 1 & 0 \\
                         0 & 0 & 1 & 0 & 0 \\
                         0 & 1 & 0 & 0 & 0 \\
                         1 & 0 & 0 & 0 & 0 \\
                         0 & 0 & 0 & 1 & 0 \\
                         0 & 0 & 0 & 0 & 0 \\
                         0 & 0 & 0 & 0 & 0 \\
\end{array}\right) .$$ Matrix ${\bf C}$ is the null matrix. Matrix ${\bf V}$ is given by $$ {\bf V}= \left( \begin{array}{cccccc}
                         1 & 0 & 0 & 0 & 0 & 1 \\
                         1 & 0 & 0 & 0 & 0 & 0 \\
                         0 & 1 & 0 & 0 & 0 & 1 \\
                         0 & 1 & 0 & 0 & 0 & 0 \\
                         0 & 0 & 1 & 0 & 0 & 1 \\
                         0 & 0 & 1 & 0 & 0 & 0 \\
                         0 & 0 & 0 & 1 & 0 & 1 \\
                         0 & 0 & 0 & 1 & 0 & 0 \\
                         0 & 0 & 0 & 0 & 1 & 1 \\
                         0 & 0 & 0 & 0 & 1 & 0 \\
\end{array}\right) ,$$ while ${\bf d}= {\mathbf 0}.$ The theoretical values for vector $\bm \lambda $ and $\bm \eta $ are

$$ \bm \lambda_0 =(\lambda_1^0, ..., \lambda_7^0)=(-3, -2, -1, 0, 1, 2, 3),\, \, \bm \eta_0 =(\eta_1^0, ..., \eta_6^0)=(0.5, 1, 1.5, 2, 2.5, 3).$$

We shall also consider different values of $a$; more concretely, we consider $a=-0.5, 0, 2/3, 1$.

For each value of $a$ we consider $R=10 000$ simulations and we reproduce the study for different sample sizes: 200, 300, 400, 500 and 1000. We must not forget that for $a=0$ and $a=1$ we have the likelihood ratio test and the chi-square ratio test statistics, respectively, but the unknown parameters are estimated using the minimum power divergence estimator with $a=2/3$ instead of the maximum likelihood estimator.

We consider as nominal size $\alpha =0.05$ and compute the {\it simulated exact  size}

$$ \hat{\alpha }_n^a:= {\sharp T_{\phi_a}^{\phi_{2/3}} > \chi^2_{g.l.;0.05} \over R}.$$

As explained in Dale (1986), we only consider the test statistics whose simulated exact size $ \hat{\alpha }_n^a$ satisfies

\begin{equation}
| logit(1-\hat{\alpha }_n^a)- logit (1-\alpha )|\leq 0.35
\end{equation}
where $logit (p)=log({p\over 1-p}).$ As a consequence, we only take under consideration the test statistics such that

\begin{equation}\label{eq23}
\hat{\alpha }_n^a \in (0.0357, 0.0695).
\end{equation}

At the same time we obtain the {\it simulated exact power} for different alternative hypothesis. More concretely, we shall consider a model with a new parameter $\lambda_8$ whose corresponding matrix ${\bf Q}_8$ is given by

$$ {\bf Q}_8=\left( \begin{array}{ccccc}
                         1 & 1 & 1 & 1 & 1  \\
                         1 & 1 & 1 & 1 & 1  \\
                         1 & 1 & 1 & 1 & 1  \\
                         1 & 1 & 1 & 1 & 1  \\
                         1 & 1 & 1 & 1 & 1  \\
                         0 & 0 & 0 & 0 & 0  \\
                         0 & 0 & 0 & 0 & 0  \\
                         0 & 0 & 0 & 0 & 0  \\
                         0 & 0 & 0 & 0 & 0  \\
                         0 & 0 & 0 & 0 & 0  \\
\end{array}\right) ,$$
and where this new parameter takes different values, namely -3, -2, -1.5, -1, -0.8, 0, 0.7, 0.9, 1, 1.3, 1.5, 2. Each of these values is related to an alternative hypothesis, except when considering value 0, that corresponds to the null hypothesis.

Simulating observations from each alternative hypothesis we get the simulated exact power for such alternatives

$$ \hat{\beta }^a:= {\sharp T_{\phi_a}^{\phi_{2/3}} > \chi^2_{g.l.;0.05} \over R}.$$ In Table 4 we present the simulated exact size as well as the simulated exact power for different values of $a$.

\begin{table}[h]
{\small
\begin{center}
\begin{tabular}{c|c|cccccccccccc}
$N$ & $a$ & -3 & -2	& -1.5 & -1 & -0.8 & 0 & 0.7 & 0.9 & 1 & 1.3 & 1.5 & 2 \\
\hline 200 & -.5 &	0.7764&	0.5242&	0.4945&	0.4786&	0.4510&	0.4041&	0.5449&	0.6352&	 0.6805&	0.8226&	0.8937&	0.9876 \\
& 0&	0.4833&	0.1662&	0.1595&	0.1505&	0.1455&	0.1095&	0.2305&	0.3232&	0.3835&	 0.6087&	0.7498&	0.9591\\
&2/3&	0.3181&	0.0617&	0.0562&	0.0472&	0.0460&	0.0354&	0.1047&	0.1769&	0.2377&	 0.4668&	0.6335&	0.9281\\
&1	&0.2915&	0.0497&	0.0426&	0.0338&	0.0351&	0.0269&	0.0885&	0.1523&	0.2124&	 0.4363&	0.6072&	0.9190\\
\hline
300 & -.5&	0.8009&	0.4473&	0.4090&	0.3724&	0.3371&	0.2680&	0.4566&	0.5838&	0.6604&	 0.8518&	0.9388&	0.9987\\
&0&	0.6202&	0.1984&	0.1789&	0.1646&	0.1410&	0.0946&	0.2519&	0.3949&	0.4891&	0.7569&	 0.8919&	0.9958\\
&2/3&	0.5094&	0.1041&	0.0876&	0.0772&	0.0627&	0.0400&	0.1599&	0.2911&	0.3818&	 0.6867&	0.8534&	0.9938\\
&1&	0.4887&	0.0870&	0.0714&	0.0621&	0.0515&	0.0331&	0.1417&	0.2687&	0.3568&	0.6714&	 0.8438&	0.9937\\
\hline 400 & -.5&	0.8311&	0.4086&	0.3823&	0.3195&	0.2828&	0.1933&	0.4200&	0.5925&	 0.6811&	0.9107&	0.9723&	0.9999\\
&0&	0.7337&	0.2245&	0.2028&	0.1712&	0.1493&	0.0819&	0.2871&	0.4761&	0.5878&	0.8739&	 0.9603&	0.9997\\
&2/3&	0.6670&	0.1429&	0.1222&	0.1006&	0.0863&	0.0429&	0.2124&	0.4027&	0.5214&	 0.8444&	0.9501&	0.9996\\
&1&	0.6531&	0.1274&	0.1099&	0.0870&	0.0733&	0.0360&	0.1995&	0.3855&	0.5052&	0.8368&	 0.9487&	0.9996\\
\hline 500 & -.5&	0.8793&	0.4129&	0.3626&	0.3065&	0.2578&	0.1448&	0.4212&	0.6309&	 0.7361&	0.9461&	0.9893&	1.0000\\
&0&	0.8207&	0.2625&	0.2293&	0.2015&	0.1626&	0.0743&	0.3357&	0.5598&	0.6781&	0.9320&	 0.9864&	1.0000\\
&2/3&	0.7821&	0.1909&	0.1600&	0.1344&	0.1075&	0.0448&	0.2755&	0.5064&	0.6388&	 0.9204&	0.9836&	1.0000\\
&1&	0.7748&	0.1769&	0.1458&	0.1196&	0.0943&	0.0397&	0.2623&	0.4962&	0.6289&	0.9177&	 0.9833&	1.0000\\
\hline 1000 & -.5&	0.9924&	0.4910&	0.5226&	0.4639&	0.3192&	0.0817&	0.6081&	0.8825&	 0.9537&	0.9999&	1.0000&	1.0000\\
&0&	0.9916&	0.3216&	0.3498&	0.3368&	0.2781&	0.0654&	0.5854&	0.8750&	0.9506&	0.9999&	 1.0000&	1.0000\\
&2/3&	0.9910&	0.2356&	0.2633&	0.2528&	0.2399&	0.0510&	0.5681&	0.8723&	0.9487&	 0.9998&	1.0000&	1.0000\\
&1&	0.9913&	0.2171&	0.2431&	0.2303&	0.2285&	0.0479&	0.5650&	0.8717&	0.9487&	0.9999&	 1.0000&	1.0000
\end{tabular}
\caption{Exact level and power for different values of $N$ and $a$.}
\end{center}
}
\end{table}

We also present the pictures for each sample size of the different alternative hypothesis for the test statistic $\lambda =-1/2, 0, 2/3, 1$ in Figures 1 to 5.

\begin{figure}[htb]
\begin{center}
$ \epsfxsize =15cm \epsfysize =9cm \epsfbox{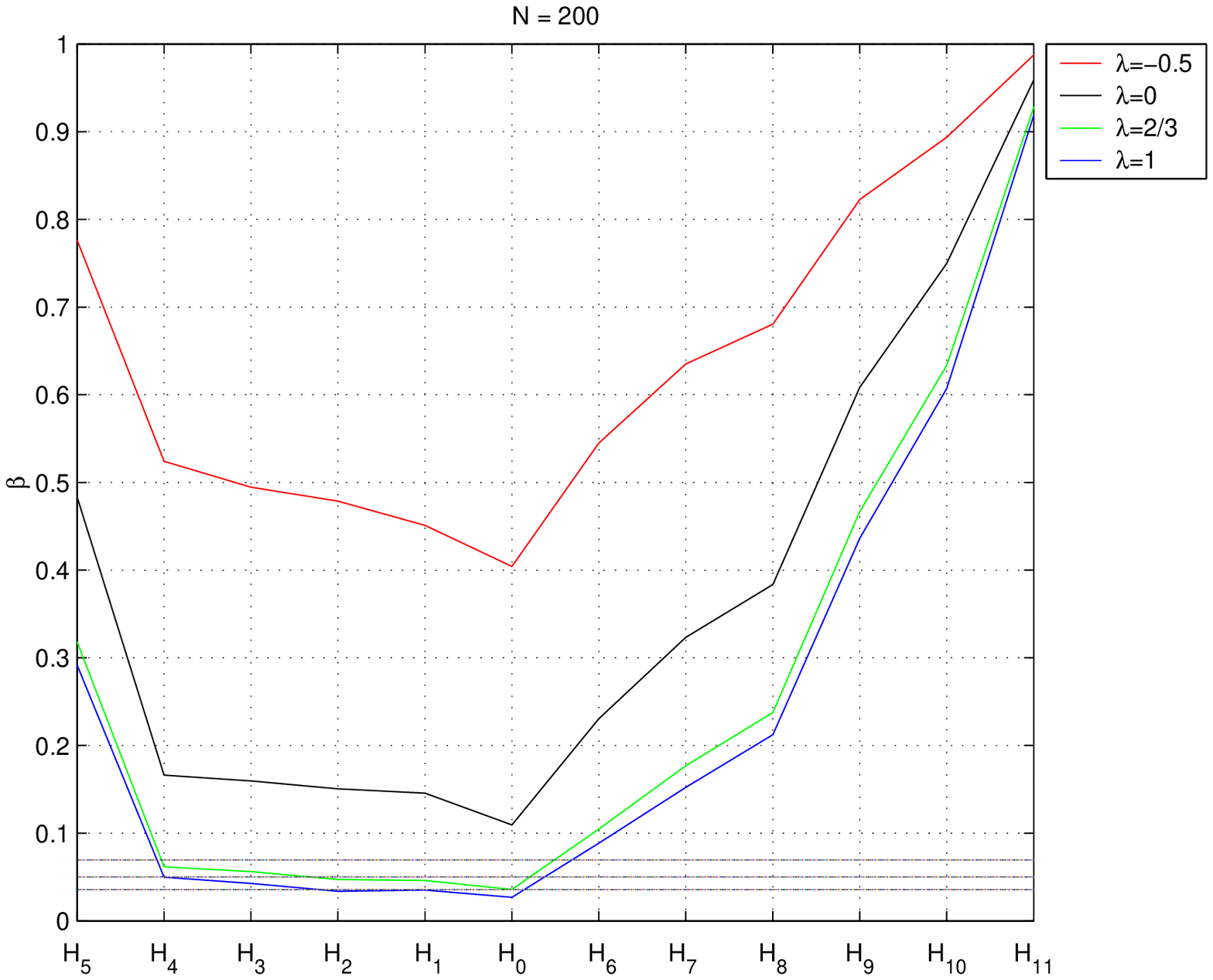} $
\end{center}
\caption{Simulated exact level and power for $N=200$.}
\end{figure}

\begin{figure}[htb]
\begin{center}
$ \epsfxsize =15cm \epsfysize =9cm \epsfbox{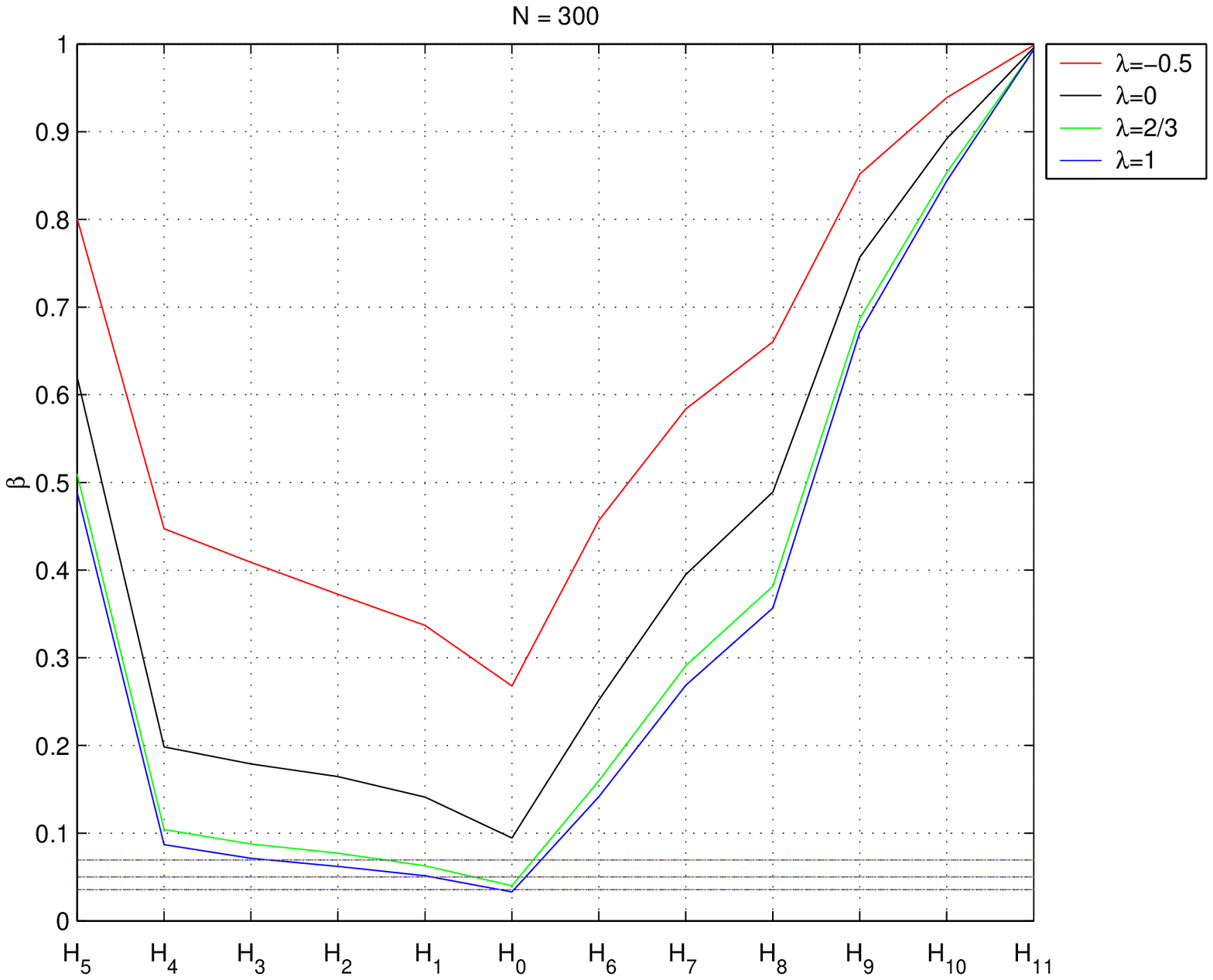} $
\end{center}
\caption{Simulated exact level and power for $N=300$.}
\end{figure}

\begin{figure}[htb]
\begin{center}
$ \epsfxsize =15cm \epsfysize =9cm \epsfbox{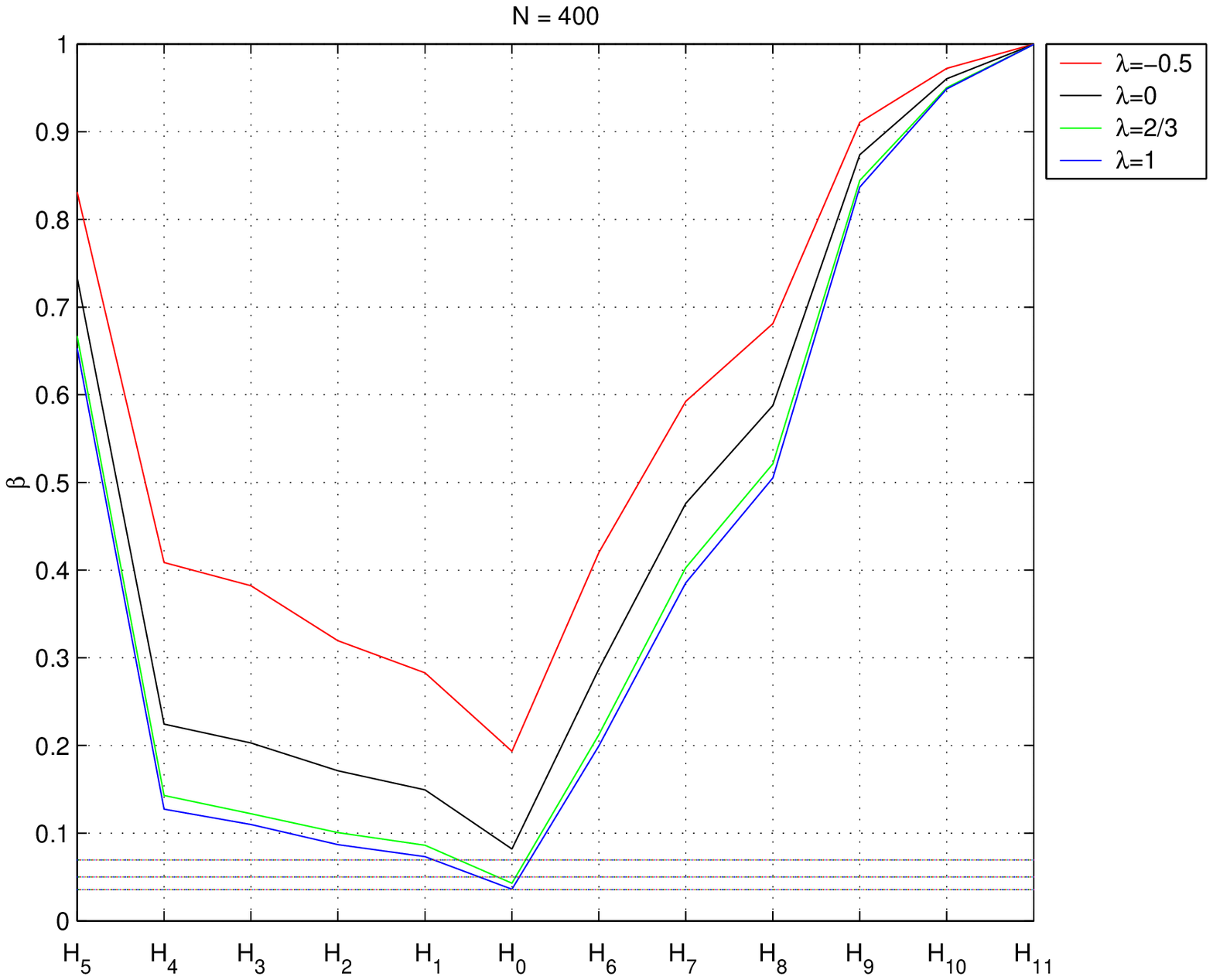} $
\end{center}
\caption{Simulated exact level and power for $N=400$.}
\end{figure}

\begin{figure}[htb]
\begin{center}
$ \epsfxsize =15cm \epsfysize =9cm \epsfbox{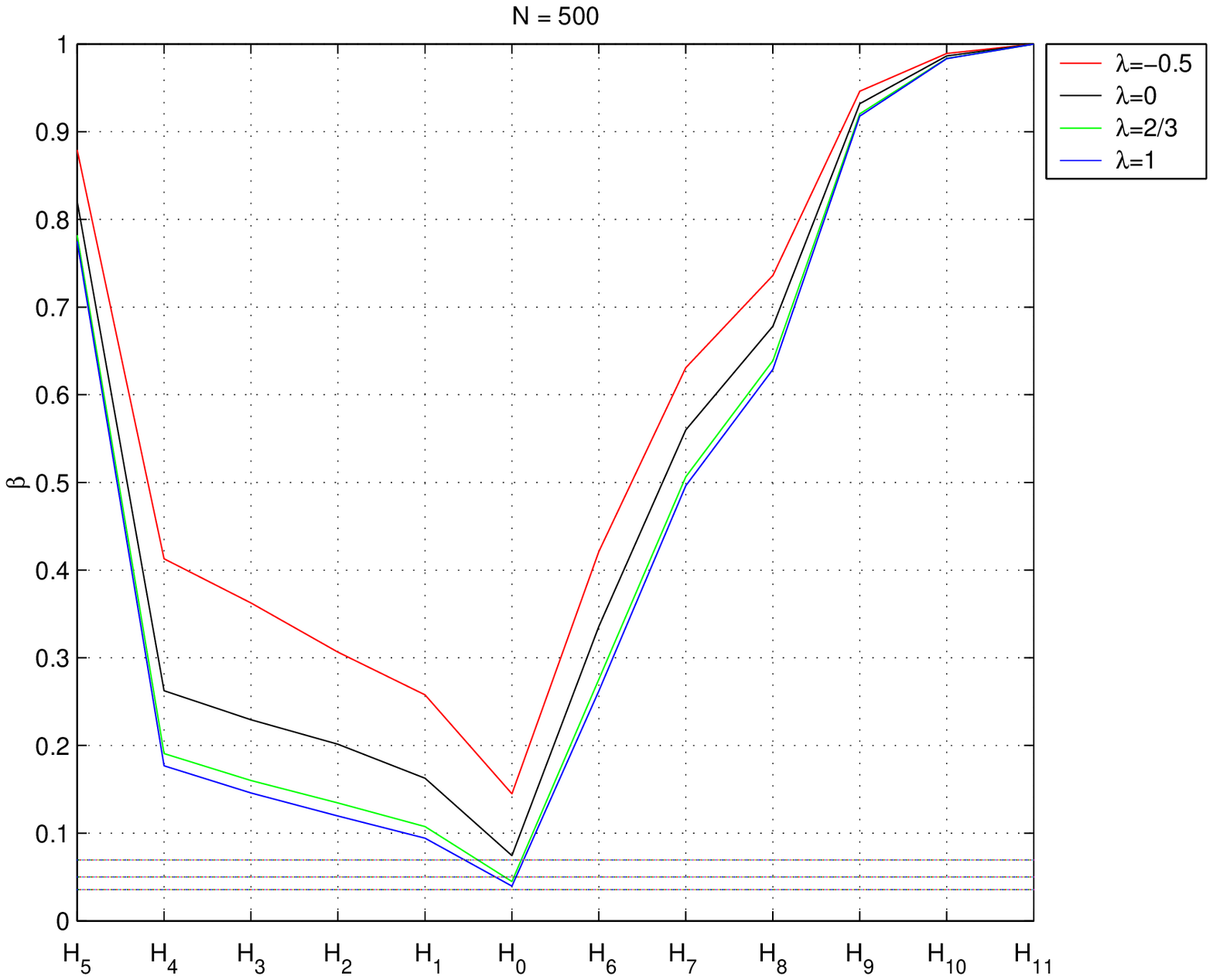} $
\end{center}
\caption{Simulated exact level and power for $N=500$.}
\end{figure}

\begin{figure}[htb]
\begin{center}
$ \epsfxsize =15cm \epsfysize =9cm \epsfbox{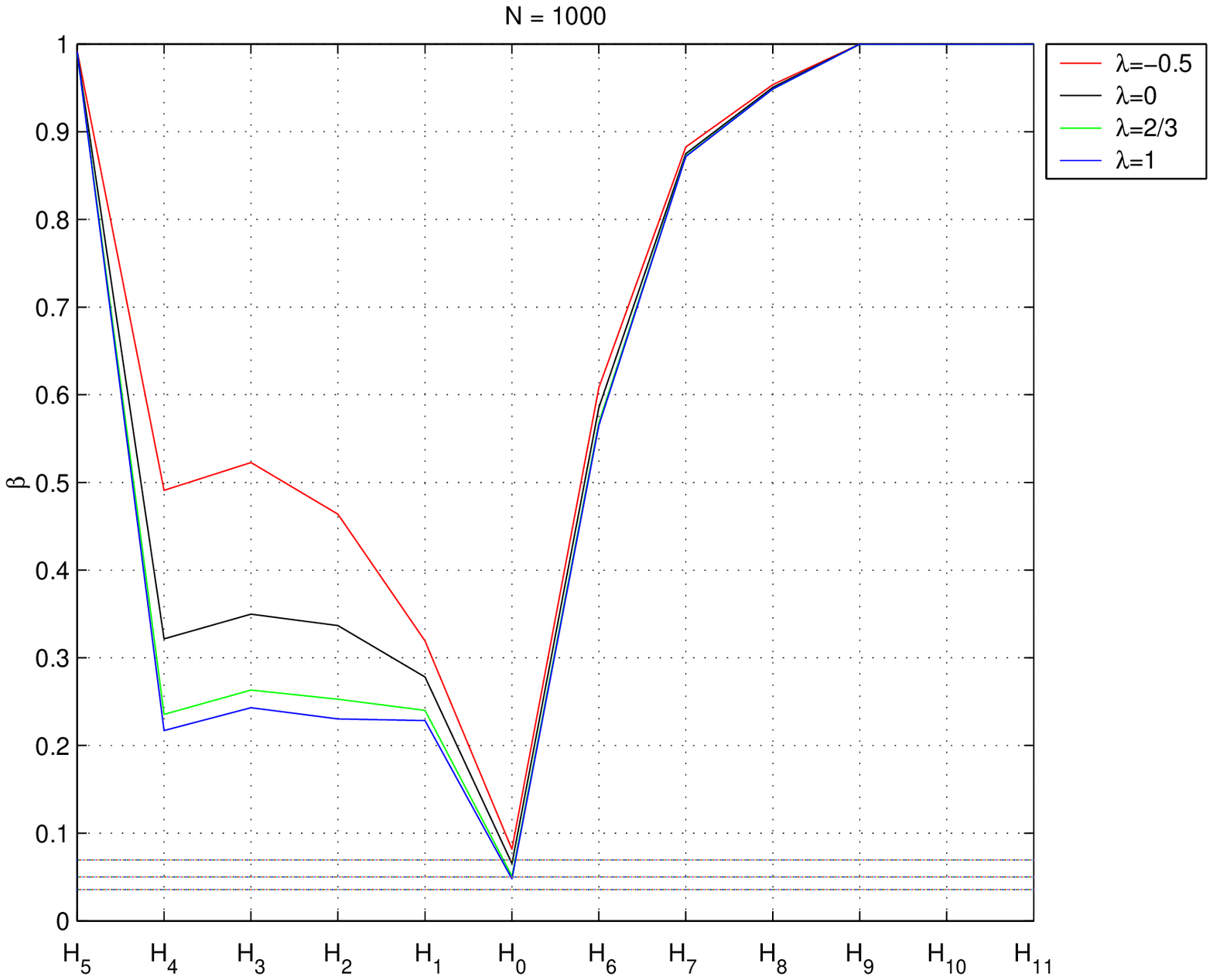} $
\end{center}
\caption{Simulated exact level and power for $N=1000$.}
\end{figure}

As it can be observed in Table 4 (see the column corresponding to $0$) and Figures 1 to 5, the simulated level is outside the interval given in \eqref{eq23} for $a=0, -0.5$ for all sample sizes under consideration; besides, for sample sizes $N=400, 500, 1000,$ the test statistic corresponding to $a=1$ lays inside this interval. Notice that the test statistic for $a=2/3$ is the only one laying in this interval for any sample size. As a straightforward conclusion, the test statistic for $a=2/3$ seems to be the best one for sample sizes $N=200, 300,$ and we just need to choose between $a=1$ and $a=2/3$ for $N=400, 500, 1000.$ For making this decision, we focus on the simulated power values, noting that they are higher for $a=2/3$ than for $a=1;$ we then conclude that $a=2/3$ seems to show a better behavior that the likelihood ratio test statistic and Pearson test statistic (with estimations obtained through $a=2/3$ instead of maximum likelihood) when dealing with LCM for binary data.

\section{Conclusions}

In this paper we have introduced phi-divergence test statistics in the context of LCM for binary data. In a previous paper, we have already shown that phi-divergence estimators can be a useful tool in this framework; now, we have treated two new problems: the problem of goodness-of-fit and the problem of selecting the best model throughout a nested sequence of models. Classically, as it can be seen for instance in Formann (1985), these problems have been considered on the basis of the likelihood-ratio-test and the chi-square test statistic. In both of them, we have derived two families of test statistics generalizing the classical ones; besides, we have obtained their asymptotical distribution under the null hypothesis of that LCM fits the data, showing that it coincides with the one of the classical test statistics; thus, they show the same behavior as the classical statistics for big sample sizes.

At this point, an interesting problem arises: are there differences for small or moderate sample sizes? To deal with this problem, we have carried out a simulation study; from this study, it seems that the phi-divergence test statistic for $a=2/3$ shows a better behavior than the classical test statistics.

%
%

\section{{\bf Acknowledgements}}

This work was partially supported by Grant MTM2012-33740.

\section{Appendix}

\textbf{Proof of Theorem 1}

A second-order Taylor expansion of $D_{\phi _{1}}\left( \boldsymbol{p},
\boldsymbol{q}\right) $ around $\left( \boldsymbol{p}\left( \boldsymbol{
\theta }_{0}\right) ,\boldsymbol{p}\left( \boldsymbol{\theta }_{0}\right)
\right) $ at $\left( \widehat{\boldsymbol{p}},\boldsymbol{p}\left( \widehat{
\boldsymbol{\theta }}_{\phi _{2}}\right) \right) $ is given by

\begin{equation*}
D_{\phi _{1}}\left( \widehat{\boldsymbol{p}},\boldsymbol{p}\left( \widehat{
\boldsymbol{\theta }}_{\phi _{2}}\right) \right) ={\phi _{1}^{\prime \prime
}(1)\over 2}\text{ }\left( \widehat{\boldsymbol{p}}-\boldsymbol{p}\left( \widehat{
\boldsymbol{\theta }}_{\phi _{2}}\right) \right) ^{T}\boldsymbol{D}_{
\boldsymbol{p}(\boldsymbol{\theta }_{0})}^{-1}\left( \widehat{\boldsymbol{p}}
-\boldsymbol{p}\left( \widehat{\boldsymbol{\theta }}_{\phi _{2}}\right)
\right) +o_{p}(N^{-1}),
\end{equation*}
being $\boldsymbol{\theta }_{0}=\left( \lambda _{1}^{0},...,\lambda
_{u}^{0},\eta _{1}^{0},...,\eta _{s}^{0}\right) .$ By $\boldsymbol{D}_{
\boldsymbol{p}(\boldsymbol{\theta }_{0})}$ we are denoting the diagonal
matrix with $\boldsymbol{p}\left( \boldsymbol{\theta }_{0}\right) $ in the
main diagonal. By Theorem 1 in Felipe et al (2014) we have
\begin{equation*}
\widehat{\boldsymbol{\theta }}_{\phi _{2}}-\boldsymbol{\theta }_{0}=\left(
\boldsymbol{L}\left( \boldsymbol{\theta }_{0}\right) ^{T}\boldsymbol{L}
\left( \boldsymbol{\theta }_{0}\right) \right) ^{-1}\boldsymbol{L}\left(
\boldsymbol{\theta }_{0}\right) ^{T}\boldsymbol{D}_{\boldsymbol{p}(
\boldsymbol{\theta }_{0})}^{-1/2}\left( \hat{\boldsymbol{p}}-\boldsymbol{p}
\left( \boldsymbol{\theta }_{0}\right) \right) +o_{p}(N^{-1/2}),
\end{equation*}
with
\begin{equation*}
\boldsymbol{L}\left( \boldsymbol{\theta }_{0}\right) =\boldsymbol{D}_{
\boldsymbol{p}(\boldsymbol{\theta }_{0})}^{-1/2}\left( \frac{\partial
\boldsymbol{p}\left( \boldsymbol{\theta }\right) }{\partial \boldsymbol{
\theta }}\right) _{\boldsymbol{\theta =\theta }_{0}}.
\end{equation*}
Therefore,
\begin{eqnarray*}
\boldsymbol{p}\left( \widehat{\boldsymbol{\theta }}_{\phi _{2}}\right) -
\boldsymbol{p}\left( \boldsymbol{\theta }_{0}\right) &=&\left( \frac{
\partial \boldsymbol{p}\left( \boldsymbol{\theta }\right) }{\partial
\boldsymbol{\theta }}\right) _{\boldsymbol{\theta =\theta }_{0}}\left( \widehat{
\boldsymbol{\theta }}_{\phi _{2}}-\boldsymbol{\theta }_{0}\right) +o_{p}(N^{-1/2}) \\
&=&\boldsymbol{D}_{\boldsymbol{p}(\boldsymbol{\theta }_{0})}^{1/2}\boldsymbol{L
}\left( \boldsymbol{\theta }_{0}\right) \left( \boldsymbol{L}\left(
\boldsymbol{\theta }_{0}\right) ^{T}\boldsymbol{L}\left( \boldsymbol{\theta }
_{0}\right) \right) ^{-1}\boldsymbol{L}\left( \boldsymbol{\theta }
_{0}\right) ^{T}\boldsymbol{D}_{\boldsymbol{p}(\boldsymbol{\theta }
_{0})}^{-1/2}\left( \hat{\boldsymbol{p}}-\boldsymbol{p}\left( \boldsymbol{
\theta }_{0}\right) \right) +o_{p}(N^{-1/2}) \\
&=&\boldsymbol{V}\left( \boldsymbol{\theta }_{0}\right) \left( \hat{
\boldsymbol{p}}-\boldsymbol{p}\left( \boldsymbol{\theta }_{0}\right) \right)
+o_{p}(N^{-1/2})
\end{eqnarray*}
with $\boldsymbol{V}\left( \boldsymbol{\theta }_{0}\right) :=\boldsymbol{D}_{
\boldsymbol{p}(\boldsymbol{\theta }_{0})}^{1/2}\boldsymbol{L}\left( \boldsymbol{
\theta }_{0}\right) \left( \boldsymbol{L}\left( \boldsymbol{\theta }
_{0}\right) ^{T}\boldsymbol{L}\left( \boldsymbol{\theta }_{0}\right) \right)
^{-1}\boldsymbol{L}\left( \boldsymbol{\theta }_{0}\right) ^{T}\boldsymbol{D}
_{\boldsymbol{p}(\boldsymbol{\theta }_{0})}^{-1/2}.$

On the other hand,

\begin{equation*}
\sqrt{N}\left( \hat{\boldsymbol{p}}-\boldsymbol{p}\left( \boldsymbol{\theta }_{0}\right) \right) \overset{\mathcal{L}}{\underset{N\longrightarrow \infty }
{\mathcal{\rightarrow }}}\mathcal{N}\left( \boldsymbol{0,\Sigma }_{
\boldsymbol{p}(\boldsymbol{\theta }_{0})}\right)
\end{equation*}
being

\begin{equation*}
\boldsymbol{\Sigma }_{\boldsymbol{p}(\boldsymbol{\theta }_{0})}=\boldsymbol{D
}_{\boldsymbol{p}(\boldsymbol{\theta }_{0})}^{{}}-\boldsymbol{p}(\boldsymbol{
\theta }_{0})\boldsymbol{p}(\boldsymbol{\theta }_{0})^{T}.
\end{equation*}

Then we have

\begin{equation*}
\widehat{\boldsymbol{p}}-\boldsymbol{p}\left( \widehat{{\bm \theta }}_{\phi _{2}}\right) =\left( {\bf I}- {\bf V}\left( {\bm \theta_{0}}\right) \right) \left( \hat{\boldsymbol{p}}-\boldsymbol{p} (\bm \theta_{0})\right) +o_{p}(N^{-1/2}) ,
\end{equation*}

and we conclude that
\begin{equation*}
\sqrt{N}\left( \widehat{\boldsymbol{p}}-\boldsymbol{p}\left( \widehat{
\boldsymbol{\theta }}_{\phi _{2}}\right) \right) \overset{\mathcal{L}}{
\underset{N\longrightarrow \infty }{\mathcal{\rightarrow }}}\mathcal{N}
\left( {\bf 0}, \left( \boldsymbol{I} - \boldsymbol{V}\left( \boldsymbol{\theta }_{0}\right) ^{T}\right) \boldsymbol{\Sigma }_{\boldsymbol{p}(\boldsymbol{
\theta }_{0})}\left( \boldsymbol{I}- \boldsymbol{V}\left( \boldsymbol{\theta }_{0}\right)
^{T}\right) \right) .
\end{equation*}

Notice that the asymptotic distribution of

\begin{equation*}
\frac{2N}{\phi _{1}^{\prime \prime }(1)}D_{\phi _{1}}\left( \widehat{
\boldsymbol{p}},\boldsymbol{p}\left( \widehat{\boldsymbol{\theta }}_{\phi
_{2}}\right) \right)
\end{equation*}
coincides with the asymptotic distribution of the quadratic form

\begin{equation*}
N \left( \widehat{\boldsymbol{p}}-\boldsymbol{p}\left( \widehat{\boldsymbol{
\theta }}_{\phi _{2}}\right) \right) ^{T}\boldsymbol{D}_{\boldsymbol{p}(
\boldsymbol{\theta }_{0})}^{-1}\left( \widehat{\boldsymbol{p}}-\boldsymbol{p}
\left( \widehat{\boldsymbol{\theta }}_{\phi _{2}}\right) \right) = \sqrt{N} \left( \widehat{\boldsymbol{p}}-\boldsymbol{p}\left( \widehat{\boldsymbol{
\theta }}_{\phi _{2}}\right) \right) ^{T}\boldsymbol{D}_{\boldsymbol{p}(
\boldsymbol{\theta }_{0})}^{-1/2}\boldsymbol{D}_{\boldsymbol{p}(
\boldsymbol{\theta }_{0})}^{-1/2}\left( \widehat{\boldsymbol{p}}-\boldsymbol{p}
\left( \widehat{\boldsymbol{\theta }}_{\phi _{2}}\right) \right) \sqrt{N}=X^TX,
\end{equation*}
with
$$X:=\sqrt{N}\boldsymbol{D}_{\boldsymbol{p}(
\boldsymbol{\theta }_{0})}^{-1/2}\left( \widehat{\boldsymbol{p}}-\boldsymbol{p}
\left( \widehat{\boldsymbol{\theta }}_{\phi _{2}}\right) \right).$$

Now, as

$$ X \overset{\mathcal{L}}{
\underset{N\longrightarrow \infty }{\mathcal{\rightarrow }}}\mathcal{N}
\left( {\bf 0}, \boldsymbol{D}_{\boldsymbol{p}(
\boldsymbol{\theta }_{0})}^{-1/2} \left( \boldsymbol{I} - \boldsymbol{V}\left( \boldsymbol{\theta }_{0}\right) ^{T}\right) \boldsymbol{\Sigma }_{\boldsymbol{p}(\boldsymbol{
\theta }_{0})}\left( \boldsymbol{I}- \boldsymbol{V}\left( \boldsymbol{\theta }_{0}\right)
^{T}\right) \boldsymbol{D}_{\boldsymbol{p}(
\boldsymbol{\theta }_{0})}^{-1/2} \right) , $$ we conclude that the asymptotic distribution of $X^TX$ will be a chi-square distribution if the matrix

$$ {\bf Q} \left( {\bm \theta }_{0}\right):=\boldsymbol{D}_{\boldsymbol{p}(
\boldsymbol{\theta }_{0})}^{-1/2} \left( \boldsymbol{I} - \boldsymbol{V}\left( \boldsymbol{\theta }_{0}\right) ^{T}\right) \boldsymbol{\Sigma }_{\boldsymbol{p}(\boldsymbol{
\theta }_{0})}\left( \boldsymbol{I}- \boldsymbol{V}\left( \boldsymbol{\theta }_{0}\right)
^{T}\right) \boldsymbol{D}_{\boldsymbol{p}(
\boldsymbol{\theta }_{0})}^{-1/2} $$
is idempotent and symmetric, and in this case de degrees of freedom will be
the trace of the matrix ${\bf Q} \left( {\bm \theta }_{0}\right) .$ Symmetry is evident. Establishing that the matrix ${\bf Q} \left( {\bm \theta }_{0}\right) $ is idempotent and that
its trace is $2^{k}-(u+t)-1$ is a simple but long and tedious exercise; a detailed proof of this fact can be found in Pardo (2006) (Theorem 6.1, pag. 259).

\textbf{Proof of Theorem 6}

Based on Theorem 1 in Felipe et al. (2014) we have
\begin{equation*}
\widehat{\boldsymbol{\theta }}_{\phi _{2}}^{A}-\boldsymbol{\theta }%
_{0}^{A}=\left( \boldsymbol{L}\left( \boldsymbol{\theta }_{0}^{A}\right) ^{T}%
\boldsymbol{L}\left( \boldsymbol{\theta }_{0}^{A}\right) \right) ^{-1}%
\boldsymbol{L}\left( \boldsymbol{\theta }_{0}^{A}\right) ^{T}\boldsymbol{D}_{%
\boldsymbol{p}(\boldsymbol{\theta }_{0}^{A})}^{-1/2}\left( \hat{\boldsymbol{p%
}}-\boldsymbol{p}\left( \boldsymbol{\theta }_{0}^{A}\right) \right)
+o_{p}(N^{-1/2}),
\end{equation*}%
being
\begin{equation*}
\boldsymbol{\theta }_{0}^{A}=\left( \boldsymbol{\theta }_{0}^{A,1},%
\boldsymbol{0},\boldsymbol{\theta }_{0}^{A,3},\boldsymbol{0}\right) ,\text{ }%
\boldsymbol{p}\left( \boldsymbol{\theta }_{0}^{A}\right) =\left( p\left(
\boldsymbol{y}_{1},\boldsymbol{\theta }_{0}^{A}\right) ,...,p\left(
\boldsymbol{y}_{2^{k}},\boldsymbol{\theta }_{0}^{A}\right) \right)
\end{equation*}%
and
\begin{equation*}
\boldsymbol{L}\left( \boldsymbol{\theta }_{0}^{A}\right) =\boldsymbol{D}_{%
\boldsymbol{p}(\boldsymbol{\theta }_{0}^{A})}^{-1/2}\left( \frac{\partial
\boldsymbol{p}\left( \boldsymbol{\theta }^{A}\right) }{\partial \boldsymbol{%
\theta }^{A,1}},\frac{\partial \boldsymbol{p}\left( \boldsymbol{\theta }%
^{A}\right) }{\partial \boldsymbol{\theta }^{A,2}},\frac{\partial
\boldsymbol{p}\left( \boldsymbol{\theta }^{A}\right) }{\partial \boldsymbol{%
\theta }^{A,3}},\frac{\partial \boldsymbol{p}\left( \boldsymbol{\theta }%
^{A}\right) }{\partial \boldsymbol{\theta }^{A,4}}\right) _{\boldsymbol{%
\theta }^{A}=\boldsymbol{\theta }_{0}^{A}}.
\end{equation*}%

Similarly,
\begin{equation*}
\widehat{\boldsymbol{\theta }}_{\phi _{2}}^{B}-\boldsymbol{\theta }%
_{0}^{B}=\left( \boldsymbol{M}\left( \boldsymbol{\theta }_{0}^{B}\right) ^{T}%
\boldsymbol{M}\left( \boldsymbol{\theta }_{0}^{B}\right) \right) ^{-1}%
\boldsymbol{M}\left( \boldsymbol{\theta }_{0}^{B}\right) ^{T}\boldsymbol{D}_{%
\boldsymbol{p}(\boldsymbol{\theta }_{0}^{B})}^{-1/2}\left( \hat{\boldsymbol{p%
}}-\boldsymbol{p}\left( \boldsymbol{\theta }_{0}^{B}\right) \right)
+o_{p}(N^{-1/2})
\end{equation*}%
being
\begin{equation*}
\boldsymbol{M}\left( \boldsymbol{\theta }_{0}^{B}\right) =\boldsymbol{D}_{%
\boldsymbol{p}(\boldsymbol{\theta }_{0}^{B})}^{-1/2}\left( \frac{\partial
\boldsymbol{p}\left( \boldsymbol{\theta }^{B}\right) }{\partial \boldsymbol{%
\theta }^{A,1}},\frac{\partial \boldsymbol{p}\left( \boldsymbol{\theta }%
^{B}\right) }{\partial \boldsymbol{\theta }^{A,3}}\right) _{\boldsymbol{%
\theta }^{B}=\boldsymbol{\theta }_{0}^{B}}.
\end{equation*}%

As
$$ \boldsymbol{\theta }_{0}^{B}= \left( \boldsymbol{\theta }_{0}^{A,1}, \boldsymbol{\theta }_{0}^{A,3}\right) ,$$
and by the hypothesis,

$$  \left( \boldsymbol{\theta }_{0}^{A,2}, \boldsymbol{\theta }_{0}^{A,4}\right) = \boldsymbol{0} ,$$
it follows that
\begin{equation*}
\widehat{\boldsymbol{\theta }}_{\phi _{2}}^{B}=\left( \widehat{\boldsymbol{%
\theta }}_{\phi _{2}}^{A,1},\widehat{\boldsymbol{\theta }}_{\phi
_{2}}^{A,3}\right) , \boldsymbol{\theta }%
_{0}^{B}=\left( \boldsymbol{%
\theta }_{0}^{A,1},\boldsymbol{\theta }_{0}^{A,3}\right)
\end{equation*}%
and
$$ \boldsymbol{p}(\boldsymbol{\theta }_{0}^{B})=\boldsymbol{p}(\boldsymbol{\theta }_{0}^{A}),$$ whence

$$ \boldsymbol{D}_{\boldsymbol{p}(\boldsymbol{\theta }_{0}^{B})} = \boldsymbol{D}_{\boldsymbol{p}(\boldsymbol{\theta }_{0}^{A})}, \, \, \boldsymbol{M}\left( \boldsymbol{\theta }_{0}^{B}\right) =\boldsymbol{M}\left( \boldsymbol{\theta }_{0}^{B}\right) =\boldsymbol{D}_{%
\boldsymbol{p}(\boldsymbol{\theta }_{0}^{A})}^{-1/2}\left( \frac{\partial
\boldsymbol{p}\left( \boldsymbol{\theta }^{A}\right) }{\partial \boldsymbol{%
\theta }^{A,1}},\frac{\partial \boldsymbol{p}\left( \boldsymbol{\theta }%
^{A}\right) }{\partial \boldsymbol{\theta }^{A,3}}\right) _{\boldsymbol{%
\theta }^{A}=\boldsymbol{\theta }_{0}^{A}}. $$

Therefore,
\begin{equation*}
\boldsymbol{p}\left( \widehat{\boldsymbol{\theta }}_{\phi _{2}}^{A}\right) -%
\boldsymbol{p}\left( \boldsymbol{\theta }_{0}^{A}\right) =\boldsymbol{D}_{%
\boldsymbol{p}(\boldsymbol{\theta }_{0}^{A})}^{1/2}\boldsymbol{L}\left(
\boldsymbol{\theta }_{0}^{A}\right) \left( \boldsymbol{L}\left( \boldsymbol{%
\theta }_{0}^{A}\right) ^{T}\boldsymbol{L}\left( \boldsymbol{\theta }%
_{0}^{A}\right) \right) ^{-1}\boldsymbol{L}\left( \boldsymbol{\theta }%
_{0}^{A}\right) ^{T}\boldsymbol{D}_{\boldsymbol{p}(\boldsymbol{\theta }%
_{0}^{A})}^{-1/2}\left( \hat{\boldsymbol{p}}-\boldsymbol{p}\left(
\boldsymbol{\theta }_{0}^{A}\right) \right) +o_{p}(N^{-1/2})
\end{equation*}%
and
{\small \begin{equation*}
\boldsymbol{p}\left( \widehat{\boldsymbol{\theta }}_{\phi _{2}}^{B}\right) -%
\boldsymbol{p}\left( \boldsymbol{\theta }_{0}^{B}\right) = \boldsymbol{p}\left( \widehat{\boldsymbol{\theta }}_{\phi _{2}}^{B}\right) -
\boldsymbol{p}\left( \boldsymbol{\theta }_{0}^{A}\right)=\boldsymbol{D}_{%
\boldsymbol{p}(\boldsymbol{\theta }_{0}^{A})}^{1/2}\boldsymbol{M}\left(
\boldsymbol{\theta }_{0}^{A}\right) \left( \boldsymbol{M}\left( \boldsymbol{%
\theta }_{0}^{A}\right) ^{T}\boldsymbol{M}\left( \boldsymbol{\theta }%
_{0}^{A}\right) \right) ^{-1}\boldsymbol{M}\left( \boldsymbol{\theta }%
_{0}^{A}\right) ^{T}\boldsymbol{D}_{\boldsymbol{p}(\boldsymbol{\theta }%
_{0}^{A})}^{-1/2}\left( \hat{\boldsymbol{p}}-\boldsymbol{p}\left(
\boldsymbol{\theta }_{0}^{A}\right) \right) +o_{p}(N^{-1/2}).
\end{equation*}}
Then,
\begin{equation*}
\boldsymbol{D}_{\boldsymbol{p}(\boldsymbol{\theta }_{0}^{A})}^{-1/2}\left(
\boldsymbol{p}\left( \widehat{\boldsymbol{\theta }}_{\phi _{2}}^{A}\right) - \boldsymbol{p}\left( \widehat{\boldsymbol{\theta }}_{\phi _{2}}^{B}\right)
\right) =\left( \boldsymbol{R}_{L}\left( \boldsymbol{\theta }_{0}^{A}\right)
-\boldsymbol{R}_{M}\left( \boldsymbol{\theta }_{0}^{A}\right) \right)
\boldsymbol{D}_{\boldsymbol{p}(\boldsymbol{\theta }_{0}^{A})}^{-1/2}\left(
\hat{\boldsymbol{p}}-\boldsymbol{p}\left( \boldsymbol{\theta }%
_{0}^{A}\right) \right) +o_{p}(N^{-1/2})
\end{equation*}%
with
\begin{equation*}
\boldsymbol{R}_{L}\left( \boldsymbol{\theta }_{0}^{A}\right) =\boldsymbol{L}%
\left( \boldsymbol{\theta }_{0}^{A}\right) \left( \boldsymbol{L}\left(
\boldsymbol{\theta }_{0}^{A}\right) ^{T}\boldsymbol{L}\left( \boldsymbol{%
\theta }_{0}^{A}\right) \right) ^{-1}\boldsymbol{L}\left( \boldsymbol{\theta
}_{0}^{A}\right) ^{T}, \text{ }\boldsymbol{R}_{M}\left( \boldsymbol{\theta }_{0}^{A}\right) =%
\boldsymbol{M}\left( \boldsymbol{\theta }_{0}^{A}\right) \left( \boldsymbol{M%
}\left( \boldsymbol{\theta }_{0}^{A}\right) ^{T}\boldsymbol{M}\left(
\boldsymbol{\theta }_{0}^{A}\right) \right) ^{-1}\boldsymbol{M}\left(
\boldsymbol{\theta }_{0}^{A}\right) ^{T}.
\end{equation*}

Therefore the asymptotic distribution of
\begin{equation*}
\sqrt{N}\boldsymbol{D}_{\boldsymbol{p}(\boldsymbol{\theta }%
_{0}^{A})}^{-1/2}\left( \boldsymbol{p}\left( \widehat{\boldsymbol{\theta }}%
_{\phi _{2}}^{A}\right) -\boldsymbol{p}\left( \widetilde{\boldsymbol{\theta }%
}_{\phi _{2}}\right) \right)
\end{equation*}%
is a normal distribution with vector mean zero and variance-covariance matrix%
\begin{equation*}
\boldsymbol{\Sigma }_{AB}=\left( \boldsymbol{R}_{L}\left( \boldsymbol{\theta
}_{0}^{A}\right) -\boldsymbol{R}_{M}\left( \boldsymbol{\theta }%
_{0}^{A}\right) \right) \boldsymbol{D}_{\boldsymbol{p}(\boldsymbol{\theta }%
_{0}^{A})}^{-1/2}\boldsymbol{\Sigma }_{\boldsymbol{p}\left( \boldsymbol{%
\theta }_{0}^{A}\right) }\boldsymbol{D}_{\boldsymbol{p}(\boldsymbol{\theta }%
_{0}^{A})}^{-1/2}\left( \boldsymbol{R}_{L}\left( \boldsymbol{\theta }%
_{0}^{A}\right) -\boldsymbol{R}_{M}\left( \boldsymbol{\theta }%
_{0}^{A}\right) \right) ,
\end{equation*}%
being
\begin{equation*}
\boldsymbol{\Sigma }_{\boldsymbol{p}\left( \boldsymbol{\theta }%
_{0}^{A}\right) }=diag\left( \boldsymbol{p}\left( \boldsymbol{\theta }%
_{0}^{A}\right) \right) -\boldsymbol{p}\left( \boldsymbol{\theta }%
_{0}^{A}\right) \boldsymbol{p}\left( \boldsymbol{\theta }_{0}^{A}\right) ^{T}
\end{equation*}%
It can be established that $\boldsymbol{\Sigma }_{AB}$ can be
written as $\boldsymbol{\Sigma }_{AB}=\boldsymbol{R}_{L}\left( \boldsymbol{%
\theta }_{0}^{A}\right) -\boldsymbol{R}_{M}\left( \boldsymbol{\theta }%
_{0}^{A}\right) $ because $\boldsymbol{R}_{L}\left( \boldsymbol{\theta }%
_{0}^{A}\right) $ and $\boldsymbol{R}_{M}\left( \boldsymbol{\theta }%
_{0}^{A}\right) $ are orthogonal projections operators and the columns of $%
\boldsymbol{M}\left( \boldsymbol{\theta }_{0}^{A}\right) $ are a subset of
the columns of $\boldsymbol{L}\left( \boldsymbol{\theta }_{0}^{A}\right) $ (see again Pardo (2006), Th. 7.1, pag. 311 for details). Then
\begin{equation*}
\boldsymbol{R}_{L}\left( \boldsymbol{\theta }_{0}^{A}\right) \boldsymbol{R}%
_{M}\left( \boldsymbol{\theta }_{0}^{A}\right) =\boldsymbol{R}_{M}\left(
\boldsymbol{\theta }_{0}^{A}\right) \boldsymbol{R}_{L}\left( \boldsymbol{%
\theta }_{0}^{A}\right) =\boldsymbol{R}_{M}\left( \boldsymbol{\theta }%
_{0}^{A}\right) .
\end{equation*}%
At the same time
\begin{equation*}
\boldsymbol{p}\left( \boldsymbol{\theta }_{0}^{A}\right) ^{1/2}\boldsymbol{R}%
_{L}\left( \boldsymbol{\theta }_{0}^{A}\right) =\boldsymbol{p}\left(
\boldsymbol{\theta }_{0}^{A}\right) ^{1/2}\boldsymbol{R}_{M}\left(
\boldsymbol{\theta }_{0}^{A}\right) =\boldsymbol{0.}
\end{equation*}%
Then we have that the matrix $\boldsymbol{R}_{L}\left( \boldsymbol{\theta }%
_{0}^{A}\right) -\boldsymbol{R}_{M}\left( \boldsymbol{\theta }%
_{0}^{A}\right) $ is symmetric and idempotent and in this case the number of
eigenvalues different to zero and equal 1 coincide with the trace of $%
\boldsymbol{R}_{L}\left( \boldsymbol{\theta }_{0}^{A}\right) -\boldsymbol{R}%
_{M}\left( \boldsymbol{\theta }_{0}^{A}\right) ,$

\begin{eqnarray*}
trace\left( \boldsymbol{R}_{L}\left( \boldsymbol{\theta }_{0}^{A}\right)
\right) &=&trace\left( \boldsymbol{L}\left( \boldsymbol{\theta }%
_{0}^{A}\right) \left( \boldsymbol{L}\left( \boldsymbol{\theta }%
_{0}^{A}\right) ^{T}\boldsymbol{L}\left( \boldsymbol{\theta }_{0}^{A}\right)
\right) ^{-1}\boldsymbol{L}\left( \boldsymbol{\theta }_{0}^{A}\right)
^{T}\right) \\
&=&trace\left( \boldsymbol{L}\left( \boldsymbol{\theta }_{0}^{A}\right) ^{T}%
\boldsymbol{L}\left( \boldsymbol{\theta }_{0}^{A}\right) \left( \boldsymbol{L%
}\left( \boldsymbol{\theta }_{0}^{A}\right) ^{T}\boldsymbol{L}\left(
\boldsymbol{\theta }_{0}^{A}\right) \right) ^{-1}\right) \\
&=&trace\left( \boldsymbol{I}_{h_{1}\times h_{2}}\right) =h_{1}.
\end{eqnarray*}

Similarly,

$$ trace\left( \boldsymbol{R}_{M}\left( \boldsymbol{\theta }_{0}^{A}\right)
\right) = h_2.$$

Finally, the second-order expansion of
\begin{equation*}
D_{\phi _{1}}\left( \boldsymbol{p}(\widehat{\boldsymbol{\theta }}_{\phi
_{2}}^{A}),\boldsymbol{p}(\widehat{\boldsymbol{\theta }}_{\phi
_{2}}^{B})\right)
\end{equation*}%
about $\left( \boldsymbol{p}\left( \boldsymbol{\theta }_{0}^{A}\right) ,%
\boldsymbol{p}\left( \boldsymbol{\theta }_{0}^{A}\right) \right) $ gives%
\begin{eqnarray*}
D_{\phi _{1}}\left( \boldsymbol{p}(\widehat{\boldsymbol{\theta }}%
_{\phi _{2}}^{A}),\boldsymbol{p}(\widehat{\boldsymbol{\theta }}_{\phi
_{2}}^{B})\right) & = & \frac{\phi _{1}^{\prime \prime }\left( 1\right) }{2}\left(
\boldsymbol{p}\left( \widehat{\boldsymbol{\theta }}_{\phi _{2}}^{A}\right) -%
\boldsymbol{p}\left( \widehat{\boldsymbol{\theta }}_{\phi _{2}}^{B} \right) \right) ^{T}\boldsymbol{D}_{\boldsymbol{p}(\boldsymbol{\theta }%
_{0}^{A})}^{-1}\left( \boldsymbol{p}\left( \widehat{\boldsymbol{\theta }}%
_{\phi _{2}}^{A}\right) -\boldsymbol{p}\left( \widehat{\boldsymbol{\theta }}_{\phi _{2}}^{B}\right) \right) +o_{p}(1) \\
& = & \frac{\phi _{1}^{\prime \prime }\left( 1\right) }{2}\left(
\boldsymbol{p}\left( \widehat{\boldsymbol{\theta }}_{\phi _{2}}^{A}\right) -%
\boldsymbol{p}\left( \widehat{\boldsymbol{\theta }}_{\phi _{2}}^{B}\right) \right) ^{T}\boldsymbol{D}_{\boldsymbol{p}(\boldsymbol{\theta }%
_{0}^{A})}^{-1/2}\boldsymbol{D}_{\boldsymbol{p}(\boldsymbol{\theta }%
_{0}^{A})}^{-1/2}\left( \boldsymbol{p}\left( \widehat{\boldsymbol{\theta }}%
_{\phi _{2}}^{A}\right) -\boldsymbol{p}\left( \widehat{\boldsymbol{\theta }}_{\phi _{2}}^{B}\right) \right) +o_{p}(1)
\end{eqnarray*}

Therefore, as we have shown in the previous proof, the asymptotic distribution of $T_{A-B}^{\phi _{1,}\phi _{2}}$ is
a chi-square distribution with $h_{1}-h_{2}$ degrees of freedom.

\end{document}